\documentclass[10pt,twoside]{article}
\usepackage{amsmath, amssymb, amsfonts, theorem, hyperref}
\def\YEAR{\year}\newcount\VOL\VOL=\YEAR\advance\VOL by-1995
\def\firstpage{1}\def\lastpage{1000}
\def\received{}\def\revised{}
\def\communicated{}

\makeatletter
\def\magnification{\afterassignment\m@g\count@}
\def\m@g{\mag=\count@\hsize6.5truein\vsize8.9truein\dimen\footins8truein}
\makeatother

\font\eightrm=cmr8
\font\caps=cmcsc10                    
\font\Caps=cmcsc10 scaled \magstep1   

\makeatletter
\@specialpagefalse\headheight=8.5pt
\def\DocMath{}
\renewcommand{\@evenhead}{%
    \ifnum\thepage>\lastpage\rlap{\thepage}\hfill%
    \else\rlap{\thepage}\slshape\leftmark\hfill{\caps\SAuthor}\hfill\fi}%
\renewcommand{\@oddhead}{%
    \ifnum\thepage=\firstpage{\DocMath\hfill\llap{\thepage}}%
    \else{\slshape\rightmark}\hfill{\caps\STitle}\hfill\llap{\thepage}\fi}%
\makeatother

\def\TSkip{\bigskip}
\newbox\TheTitle{\obeylines\gdef\GetTitle #1
\ShortTitle  #2
\SubTitle    #3
\Author      #4
\ShortAuthor #5
\EndTitle
{\setbox\TheTitle=\vbox{\baselineskip=20pt\let\par=\cr\obeylines%
\halign{\centerline{\Caps##}\cr\noalign{\medskip}\cr#1\cr}}%
	\copy\TheTitle\TSkip\TSkip%
\def\next{#2}\ifx\next\empty\gdef\STitle{#1}\else\gdef\STitle{#2}\fi%
\def\next{#3}\ifx\next\empty%
    \else\setbox\TheTitle=\vbox{\baselineskip=20pt\let\par=\cr\obeylines%
    \halign{\centerline{\caps##} #3\cr}}\copy\TheTitle\TSkip\TSkip\fi%
\centerline{\caps #4}\TSkip\TSkip%
\def\next{#5}\ifx\next\empty\gdef\SAuthor{#4}\else\gdef\SAuthor{#5}\fi%
\ifx\received\empty\relax
    \else\centerline{\eightrm Received: \received}\fi%
\ifx\revised\empty\TSkip%
    \else\centerline{\eightrm Revised: \revised}\TSkip\fi%
\ifx\communicated\empty\relax
    \else\centerline{\eightrm Communicated by \communicated}\fi\TSkip\TSkip%
\catcode'015=5}}\def\Title{\obeylines\GetTitle}
\def\Abstract{\begingroup\narrower
    \parskip=\medskipamount\parindent=0pt{\caps Abstract. }}
\def\EndAbstract{\par\endgroup\TSkip}

\long\def\MSC#1\EndMSC{\def\arg{#1}\ifx\arg\empty\relax\else
     {\par\narrower\noindent%
     2000 Mathematics Subject Classification: #1\par}\fi}

\long\def\KEY#1\EndKEY{\def\arg{#1}\ifx\arg\empty\relax\else
	{\par\narrower\noindent Keywords and Phrases: #1\par}\fi\TSkip}

\newbox\TheAdd\def\Addresses{\vfill\copy\TheAdd\vfill
    \ifodd\number\lastpage\vfill\eject\phantom{.}\vfill\eject\fi}
{\obeylines\gdef\GetAddress #1
\Address #2 
\Address #3
\Address #4
\EndAddress
{\def\xs{4.3truecm}\parindent=0pt
\setbox0=\vtop{{\obeylines\hsize=\xs#1\par}}\def\next{#2}
\ifx\next\empty 
     \setbox\TheAdd=\hbox to\hsize{\hfill\copy0\hfill}
\else\setbox1=\vtop{{\obeylines\hsize=\xs#2\par}}\def\next{#3}
\ifx\next\empty 
     \setbox\TheAdd=\hbox to\hsize{\hfill\copy0\hfill\copy1\hfill}
\else\setbox2=\vtop{{\obeylines\hsize=\xs#3\par}}\def\next{#4}
\ifx\next\empty\ 
     \setbox\TheAdd=\vtop{\hbox to\hsize{\hfill\copy0\hfill\copy1\hfill}
                \vskip20pt\hbox to\hsize{\hfill\copy2\hfill}}
\else\setbox3=\vtop{{\obeylines\hsize=\xs#4\par}}
     \setbox\TheAdd=\vtop{\hbox to\hsize{\hfill\copy0\hfill\copy1\hfill}
	        \vskip20pt\hbox to\hsize{\hfill\copy2\hfill\copy3\hfill}}
\fi\fi\fi\catcode'015=5}}\gdef\Address{\obeylines\GetAddress}

\newcommand{\x}{\mathbf{x}}
\newcommand{\X}{\mathbf{X}}
\newcommand{\y}{\mathbf{y}}

\newcommand{\p}{\mathbf{P}}
\newcommand{\al}{\boldsymbol{\alpha}}

\newcommand{\rxy}{|\x- \y|}

\newcommand{\rhs}{$\mathrm{r.\,h.\,s.}$ }

\newtheorem{theorem}{Theorem}
\newtheorem{lemma}[theorem]{Lemma}
\newtheorem{proposition}[theorem]{Proposition}
\newtheorem{remark}[theorem]{Remark}
\newtheorem{assumption}[theorem]{Assumption}
\newtheorem{corollary}[theorem]{Corollary}
\newenvironment{proof}[1]{\noindent{\bf Proof{#1}.}}{$\bullet$}

\DeclareMathOperator{\Dom}{Dom}
\DeclareMathOperator{\supp}{supp}

\makeatletter
\@addtoreset{equation}{section}
\makeatother

\begin{document}
\Title Essential Spectrum of Multiparticle
Brown--Ravenhall Operators in External Field
\ShortTitle Essential Spectrum of Brown--Ravenhall Operators
\SubTitle   
\Author Sergey Morozov
\ShortAuthor 
\EndTitle
\Abstract The essential spectrum of multiparticle Brown--Ravenhall operators is characterized in terms of two--cluster decompositions for a wide class of external fields and interparticle interactions and for the systems with prescribed symmetries.
\EndAbstract
\MSC 81V55, 81Q10
\EndMSC
\KEY 
\EndKEY
\Address Mathematisches~Institut LMU M\"unchen
Theresienstr. 39
80333 Munich, Germany
\Address
\Address
\Address
\EndAddress
\section{Introduction}

It is well known that the eigenvalues of the one--particle Dirac operator are in much better accordance with the spectroscopic data then the eigenvalues of the Schr\"odinger operator. However, due to the presence of the negative continuum of positronic states the many--particle Coulomb--Dirac operator has no eigenvalues and its essential spectrum is the whole real line. Coupling with the quantized electromagnetic field does not correct this situation. However, there are ways to construct a \emph{semibounded} operator which will take the relativistic effects into account. Such models, although nonlocal, find their applications in numerical studies of heavy elements and cosmology, where the relativistic effects cannot be ignored.

The most obvious choice of the kinetic energy (sometimes called Chandra\-sek\-har or Herbst operator) given by $\sqrt{\mathbf p^2c^2+ m^2c^4}$, $\mathbf p$ and $m$ being the momentum and mass of the particle, suffers from the lack of semiboundedness for nuclear charges exceeding $87$, as shown in \cite{Herbst1977}. Most other operators considered in the literature are obtained by reducing the (multiparticle) Dirac operator onto some subspace on which it becomes semibounded. One of such models, extensively studied recently, is by Brown and Ravenhall \cite{BrownRavenhall1951}, see also Bethe and Salpeter \cite{BetheSalpeter1957}, Sucher \cite{Sucher1980, Sucher1987}. In this model we require that every particle stays in the positive spectral subspace of the \emph{free} Dirac operator. Since the multiplication by interaction potentials does not leave this subspace invariant, the potential energy terms should be projected back by the corresponding projector.

The mathematical study of the Brown--Ravenhall operator started from the one--particle case in the article of Evans, Perry, and Siedentop \cite{EvansPerrySiedentop1996}. The authors have proved that the atomic Hamiltonian is semibounded from below for nuclear charges not exceeding $124$. This makes the Brown--Ravenhall model applicable to all existing elements. It was also proved in \cite{EvansPerrySiedentop1996} that the essential spectrum of the one--particle atomic Brown--Ravenhall operator is $[mc^2, \infty)$ with $m$ being the mass of the particle, and that the singular continuous spectrum is empty.

Further studies of the Brown--Ravenhall operator include the improved lower bounds by Tix \cite{Tix1997, Tix1998} (see also Burenkov and Evans \cite{BurenkovEvans1998}) in the atomic case, the proof that the eigenvalues of Brown--Ravenhall operator are strictly bigger than those of the one--particle Dirac operator by Griesemer et al. \cite{GriesemerLewisSiedentop1999}, proofs of stability of one-electron molecule by Balinsky and Evans \cite{BalinskyEvans1999} and the proof of stability of matter by Hoever and Siedentop \cite{HoeverSiedentop1999}.
The essential spectrum of the many--particle operator was characterized by Jakuba\ss a--Amundsen \cite{Jakubassa2005, Jakubassa2007}, and Morozov and Vugalter \cite{MorozovVugalter2006} in terms of two--cluster decompositions. This is usually referred to as HVZ theorem after the well known result for the many particle Schr\"odinger oprator. It is also shown in \cite{MorozovVugalter2006} that the neutral atoms or positively charged atomic ions have infinitely many bound states.

In all these previous studies the nuclei were considered as fixed sources of the external field, the particles were assumed to be identical, and the interaction potentials were purely Coulombic. 

In this paper we generalize the HVZ theorem of \cite{Jakubassa2005, Jakubassa2007, MorozovVugalter2006} as follows: We allow any number of (massive) particles of the system to be identical. We allow quite general matrix interaction potentials. In particular, our result applies in the presence of the magnetic fields if the vector potential decays at infinity in some weak sense. Another problem we address is the reduction to any irreducible representations of the groups of rotation--reflection symmetry and permutations of identical particles. Note that such a reduction allows to analyze the eigenvalues of some irreducible representations even if they are embedded into the continuous spectrum of another representation. Existence of such embedded eigenvalues is well known for atomic and molecular Schr\"odinger operators.

From the technical point of view, the nonlocality of the model due to the presence of the spectral projections of the free Dirac operator is overcome with the same ideas as in \cite{MorozovVugalter2006}. One more complication should be stressed: for the Brown--Ravenhall operator the center of mass motion cannot be separated in the same way as it is usually done for Schr\"odinger operators, where the complete Hamiltonian without external field can be represented in suitable coordinates as
\[
\mathcal{H}= A\otimes I+ I\otimes B,
\]
where $A$ describes the free motion of the center of mass and $B$ is the internal Hamiltonian of the system (see \cite{JorgensWeidmann1973}). Such a decomposition appears to be especially fruitful in the presence of rotation symmetries. Since it cannot be obtained for pseudorelativistic operators due to the form of kinetic energy, we have used completely different approach based on the commutation of the Hamiltonian with the \emph{absolute value} of the total momentum of the system.

Note that the proof of the HVZ theorem for a system of particles described by the Chandrasekhar operator, was till now not known in presence of rotation--reflection symmetries (see the article of Lewis, Siedentop and Vugalter \cite{LewisSiedentopVugalter1997} for the proof without symmetries). Such a proof can now be obtained as a simplified modification of the proof given in this paper.

In Section~\ref{setup section} we introduce the model and make the necessary assumptions. At the end of this section we formulate the main result in Theorem~\ref{HVZ theorem}. The rest of the article contains the proof of this theorem.

\section{Setup and Main Result}\label{setup section}

$[A, B]= AB- BA$ is the commutator of two operators.
$\langle\cdot, \cdot\rangle$ and $\|\cdot\|$ stand for the inner product and the norm in $L_2(\mathbb{R}^{3d}, \mathbb{C}^{4^d})$, where $d$ is the dimension of the underlying configuration space. Irrelevant constants are denoted by $C$. $I_\Omega$ is the indicator function of the set $\Omega$. For a selfadjoint operator $A$ we denote its spectrum and the corresponding sesquilinear form by $\sigma(A)$ and $\langle A\cdot, \cdot\rangle= \langle\cdot, A\cdot\rangle$, respectively. We use the conventional units $\hbar= c= 1$. Sometimes we denote the unitary Fourier transform by $\widehat\cdot$.

In the Hilbert space $L_2(\mathbb{R}^3, \mathbb{C}^4)$ the Dirac operator describing a particle of mass $m> 0$ is given by
\begin{equation*}\label{Dirac}
D_m= -i\al\cdot\nabla+ \beta m,
\end{equation*}
where $\al:= (\alpha_1, \alpha_2, \alpha_3)$ and $\beta$ are the $4\times 4$ Dirac matrices \cite{Thaller1992}.
The form domain of $D_m$ is the Sobolev space $H^{1/2}(\mathbb{R}^3, \mathbb{C}^4)$ and the spectrum is $(-\infty, -m]\cup[m, +\infty)$. Let $\Lambda_m$ be the projector onto the positive spectral subspace of $D_m$:
\begin{equation}\label{Lambda_m}
\Lambda_m:=\frac{1}{2}+ \frac{-i\al\cdot\nabla+ \beta m}{2\sqrt{-\Delta+ m^2}}.
\end{equation}
We consider a finite system of $N$ particles with positive masses $m_n, \: n= 1, \dots, N$. To simplify the notation we write $D_n$ and $\Lambda_n$ for $D_{m_n}$ and $\Lambda_{m_n}$, respectively.
Let $\mathfrak{H}_N:= \underset{n= 1}{\overset{N}{\otimes}}\Lambda_nL_2(\mathbb{R}^3, \mathbb{C}^4)$ be the Hilbert space with the inner product induced by those of $\underset{n= 1}{\overset{N}{\otimes}}L_2(\mathbb{R}^3, \mathbb{C}^4)\cong L_2(\mathbb{R}^{3N}, \mathbb{C}^{4^N})$.
In this space the $N$--particle Brown--Ravenhall operator is formally given by
\begin{equation}\label{H_N'}
\mathcal{H}_N= \Lambda^N\bigg(\underset{n= 1}{\overset{N}{\sum}}(D_n+ V_n)+ \underset{n< j}{\overset{N}{\sum}}U_{nj}\bigg)\Lambda^N,
\end{equation}
with
\begin{equation}\label{Lambda^N}
\Lambda^N:= \prod_{n= 1}^N\Lambda_n= \underset{n= 1}{\overset N\otimes}\Lambda_n.
\end{equation}
Here and below the indices $n$ and $j$ indicate the particle, on whose coordinates the corresponding operator acts.
In \eqref{H_N'} $V_n$ is the external field potential for the $n^{th}$ particle, i.e., the operator of multiplication by a hermitian $4\times 4$ matrix--function $V_n(\x_n)$, $n= 1, \dots, N$, and $U_{nj}$ is the potential energy of the interaction between the $n^{th}$ and $j^{th}$ particles, given by the operator of multiplication by a hermitian $16\times 16$ matrix--function $U_{nj}(\x_n- \x_j)$, $n< j= 1, \dots, N$.
More explicitly, if we let $s_j\in \{1, 2, 3, 4\}$ be the spinor index of the $j^{th}$ particle, then
\begin{equation*}\begin{split}\label{explicit V_n}
&(V_n\psi)(\x_1, s_1; \dots; \x_n, s_n; \dots; \x_N, s_N)\\ &:= \sum_{\widetilde s_n}V_n^{s_n, \widetilde s_n}(\x_n)\psi(\x_1, s_1; \dots; \x_n, \widetilde s_n; \dots; \x_N, s_N),
\end{split}\end{equation*}
and
\begin{equation*}\label{explicit U_nj}\begin{split}
&(U_{nj}\psi)(\x_1, s_1; \dots; \x_n, s_n; \dots; \x_j, s_j; \dots; \x_N, s_N)\\ &:= \sum_{\widetilde s_n, \widetilde s_j}U_{nj}^{s_ns_j, \widetilde s_n\widetilde s_j}(\x_n- \x_j)\psi(\x_1,  s_1; \dots; \x_n, \widetilde s_n; \dots; \x_j, \widetilde s_j; \dots; \x_N,  s_N).
\end{split}\end{equation*}

Before we make other assumptions on the interaction potentials, let us consider possible decompositions of the system into two clusters. Let $Z= (Z_1, Z_2)$ be a decomposition of the index set $I:= \{1, \dots, N\}$ into two disjoint subsets:
\begin{equation*}\label{decomposition}
I= Z_1\cup Z_2, \quad Z_1\cap Z_2= \varnothing.
\end{equation*}
Let
\begin{equation}\label{N_j}
N_j:= \#Z_j, \quad j= 1, 2
\end{equation}
be the number of particles in each cluster. We will write $n\# j$ if $n$ and $j$ belong to different clusters.
Let
\begin{equation}\label{H_Z_1}
\mathcal{H}_{Z, 1}:= \sum_{n\in Z_1}(D_n+ V_n)+ \sum_{\substack{n,j\in Z_1\\ n< j}}U_{nj},
\end{equation}
\begin{equation}\label{H_Z_2}
\mathcal{H}_{Z, 2}:= \sum_{n\in Z_2}D_n+ \sum_{\substack{n,j\in Z_2\\ n< j}}U_{nj}.
\end{equation}
We omit $\mathcal{H}_{Z, j}$ if $Z_j= \varnothing$, $j= 1, 2$.
Let us introduce the operators corresponding to noninteracting clusters, with the second cluster transferred far away from the sources of the external field:
\begin{equation}\label{reduced spaces}
\widetilde{\mathcal{H}}_{Z, j}:= \Lambda_{Z, j}\mathcal{H}_{Z, j}\Lambda_{Z, j},\quad \textrm{in}\quad \mathfrak{H}_{Z, j}:= \underset{n\in Z_j}{\otimes}\Lambda_nL_2(\mathbb{R}^3, \mathbb{C}^4), \quad j= 1, 2,
\end{equation}
where
\begin{equation*}\label{Lambda_Zj}
\Lambda_{Z, j}:= \underset{n\in Z_j}{\prod}\Lambda_n.
\end{equation*}

We make the following assumptions:

\begin{assumption}\label{upper assumption}
There exists $C>  0$ such that for any $Z$ and $j= 1, 2$
\begin{equation}\label{upper estimate}
\big|\langle\mathcal{H}_{Z, j}\varphi, \psi\rangle\big|\leqslant C\|\varphi\|_{H^{1/2}}\|\psi\|_{H^{1/2}},\quad \textrm{for any}\quad \varphi, \psi\in\underset{n\in Z_j}{\otimes}H^{1/2}(\mathbb{R}^3, \mathbb{C}^4).
\end{equation}
\end{assumption}
For Coulomb interaction potentials \eqref{upper estimate} follows from Kato's inequality.

\begin{assumption}\label{lower assumption}
There exist $C_1> 0$ and $C_2\in \mathbb{R}$ such that for any $Z$
\begin{equation}\label{assumption 1}\begin{split}
\langle\widetilde{\mathcal{H}}_{Z, j}\psi, \psi\rangle\geqslant C_1\langle\sum_{n\in Z_j}D_n\psi, \psi\rangle- C_2\|\psi\|^2&, \\ \textrm{for any}\quad \psi\in\underset{n\in Z_j}{\otimes}\Lambda_nH^{1/2}(\mathbb{R}^3, \mathbb{C}^4)&, \quad j= 1, 2.
\end{split}\end{equation}
\end{assumption}

\begin{remark}\label{equivalence remark}
Note that for $\psi\in\underset{n\in Z_j}{\otimes}\Lambda_nH^{1/2}(\mathbb{R}^3, \mathbb{C}^4)$ the metric
\[
\langle\sum\limits_{n\in Z_j}D_n\psi, \psi\rangle^{1/2}= \big\|\sum_{n\in Z_j}|D_n|^{1/2}\psi\big\|
\]
is equivalent to the norm of $\psi$ in $\underset{n\in Z_j}{\otimes}H^{1/2}(\mathbb{R}^3, \mathbb{C}^4)$, since
\begin{equation}\label{square root formula}
\Lambda_nD_n\Lambda_n= \Lambda_n|D_n|\Lambda_n= \Lambda_n\sqrt{-\Delta+ m_n^2}\Lambda_n.
\end{equation}
\end{remark}
An equivalent formulation of Assumption~\ref{lower assumption} is that the operator $\widetilde{\mathcal{H}}_{Z, j}$ is semibounded from below even if we multiply all the interaction potentials by $1+ \varepsilon$ with $\varepsilon> 0$ small enough. This is only slightly more restrictive than the semiboundedness of $\widetilde{\mathcal{H}}_{Z, j}$.

\begin{assumption}\label{square integrability assumption}
For any $R> 0$ there exists a finite constant $C_R\geqslant 0$ such that
\begin{equation}\label{V and U integrability}
\sum_{n= 1}^N\bigg(\int_{|\x|\leqslant R}\big|V_n(\x)\big|^2d\x\bigg)^{1/2}+ \sum_{n< j}^N\bigg(\int_{|\x|\leqslant R}\big| U_{nj}(\x)\big|^2d\x\bigg)^{1/2}\leqslant C_R.
\end{equation}
\end{assumption}
This means that the interaction potentials are locally square integrable.

\begin{assumption}\label{decay assumption}
For any $\varepsilon> 0$ there exists $R> 0$ big enough such that for all $n< j= 1, \dots, N$
\begin{equation}\label{V_n decay}
\|V_n I_{\{|\x_n|> R\}}\psi\|\leqslant \varepsilon\big\||D_n|^{1/2}\psi\big\|, \quad\textrm{for all}\quad \psi\in H^{1/2}(\mathbb{R}^3, \mathbb{C}^4),
\end{equation}
and
\begin{equation}\label{U_nj decay}\begin{split}
\|U_{nj} I_{\{|\x_n- \x_j|> R\}}\varphi\|\leqslant \varepsilon\min\Big\{\big\||D_n|^{1/2}\varphi\big\|, \big\||D_j|^{1/2}\varphi\big\|\Big\},\\ \textrm{for all}\quad \varphi\in H^{1/2}(\mathbb{R}^6, \mathbb{C}^{16}).
\end{split}\end{equation}
\end{assumption}
By Remark~\ref{equivalence remark} this assumption is \emph{weaker} then the decay of $L_\infty$ norms of the interaction potentials at infinity.

It follows from \eqref{assumption 1} and Remark~\ref{equivalence remark} that for any $Z$ there exists a constant $C> 0$ such that for any $\psi\in\underset{n\in Z_j}{\otimes}\Lambda_nH^{1/2}(\mathbb{R}^3, \mathbb{C}^4)$
\begin{equation}\label{backward equivalence}
\|\psi\|_{H^{1/2}}^2\leqslant C\big(\langle\widetilde{\mathcal{H}}_{Z, j}\psi, \psi\rangle+ \|\psi\|^2\big), \quad j= 1, 2.
\end{equation}
Hence by Assumptions~\ref{upper assumption} and \ref{lower assumption}, the quadratic forms of operators \eqref{reduced spaces} (and, in particular, $\mathcal H_N$) are semibounded from below and closed on $\underset{n\in Z_j}{\otimes}\Lambda_nH^{1/2}(\mathbb{R}^3, \mathbb{C}^4)$. Thus these operators are well--defined in the form sense.

Some particles of the system (say, $k^{th}$ and $l^{th}$) can be identical (in which case $m_k= m_l$, $V_k= V_l$, and $U_{kj}= U_{lj}$ for all $j$). Then the operator $\mathcal H_N$ can be reduced to the subspace of functions which transform in a certain way under permutations of identical particles. The most physically motivated assumption is that any transposition of two identical particles should change the sign of the wave function $\psi\in\mathfrak{H}_N$ describing the system. This is the Pauli principle applied to the identical fermions (the model describes spin $1/2$ particles, thus fermions).

Let $\Pi$ be the subgroup of the symmetric group $\mathcal{S}_N$ generated by transpositions of identical particles. We denote the number of elements of $\Pi$ by $h_\Pi$. Let $E$ be some irreducible representation of $\Pi$ with dimension $d_E$ and character $\xi_E$. For $\psi\in \mathfrak{H}_N$ let
\begin{equation}\label{permutation projection}
P^E\psi:= \frac{d_E}{h_\Pi}\sum_{\pi\in \Pi}\overline{\xi_E(\pi)}\pi\psi,
\end{equation}
where $\pi$ is the operator of permutation:
\begin{equation*}\label{pi_s}
(\pi\psi)(\x_1,  s_1; \dots; \x_N,  s_N)= \psi(\x_{\pi^{-1}(1)},  s_{\pi^{-1}(1)}; \dots; \x_{\pi^{-1}(N)},  s_{\pi^{-1}(N)}).
\end{equation*}
Here $ s_1, \dots,  s_N$ are the spinor coordinates of the particles. The operator $P^E$ defined in \eqref{permutation projection} is the projector to the subspace of functions in $\mathfrak{H}$ which transform according to the representation $E$ of $\Pi$. Since any $\pi\in\Pi$ commutes with $\mathcal{H}_N$, $P^E$ reduces $\mathcal{H}_N$. Let $\mathcal{H}_N^E$ be the corresponding reduced selfadjoint operator in 
\begin{equation*}\label{H_N^E}
\mathfrak{H}_N^E:= P^E\mathfrak{H}_N.
\end{equation*}

For a decomposition $Z= (Z_1, Z_2)$ let $\Pi^Z_j$ be the group generated by transpositions of identical particles inside $Z_j$, $j= 1, 2$. For any irreducible representation $E_j$ of $\Pi^Z_j$ with dimension $d_{E_j}$ and character $\xi_{E_j}$ the projection to the space of functions in $\mathfrak{H}_j$ transforming according to $E_j$ under action of $\Pi^Z_j$ is given by
\begin{equation*}\label{partial permutation projection}
P^{E_j}\psi:= \frac{d_{E_j}}{h_{\Pi^Z_j}}\sum_{\pi\in \Pi^Z_j}\overline{\xi_{E_j}(\pi)}\pi\psi, \quad \psi\in\mathfrak{H}_{Z, j},
\end{equation*}
where $h_{\Pi^Z_j}$ is the cardinality of $\Pi^Z_j$. Projectors $P^{E_j}$ reduce operators $\widetilde{\mathcal{H}}_{Z, j}$. We introduce the reduced operators $\widetilde{\mathcal{H}}_{Z, j}^{E_j}$ in
\begin{equation*}\label{H_Z_j^E}
\mathfrak{H}_{Z, j}^{E_j}:= P^{E_j}\mathfrak{H}_{Z, j}, \quad j= 1, 2.
\end{equation*}
Given an irreducible representation $E$ of $\Pi$ and a decomposition $Z= (Z_1, Z_2)$, we have
\begin{equation}\label{symmetry embedding}
\mathfrak{H}_N^E\subset \underset{(E_1, E_2)}{\oplus}\big(\mathfrak{H}_{Z, 1}^{E_1}\otimes\mathfrak{H}_{Z, 2}^{E_2}\big),
\end{equation}
where $E_{1,2}$ are some irreducible representations of $\Pi^Z_j$. We write $(E_1, E_2)\underset{Z}{\prec} E$ if the corresponding term cannot be omitted on the \rhs of \eqref{symmetry embedding} without violation of the inclusion.

Apart from permutations of identical particles the operator $\mathcal{H}_N^E$ can have some rotation--reflection symmetries. Let $\gamma$ be an orthogonal transform in $\mathbb{R}^3$: the rotation around the axis directed along a unit vector $\mathbf{n}_\gamma$ through an angle $\varphi_\gamma$, possibly combined with the reflection $\x\mapsto -\x$. The corresponding unitary operator $O_\gamma$ acts on the functions $\psi\in\mathfrak{H}^N$ as (see \cite{Thaller1992}, Chapter 2)
\begin{equation*}\label{rotation}
(O_\gamma\psi)(\x_1, \dots, \x_N)= \prod_{n= 1}^Ne^{-i\varphi_\gamma\mathbf{n}_\gamma\cdot\mathbf{S}_n}\psi(\gamma^{-1}\x_1, \dots, \gamma^{-1}\x_N).
\end{equation*}
Here $\mathbf{S}_n= - \frac i4\alpha_n\wedge\alpha_n$ is the spin operator acting on the spinor coordinates of the $n^{th}$ particle.
The compact group of orthogonal transformations $\gamma$ such that $O_\gamma$ commutes with $V_n$ and $U_{nj}$ for all $n,\,j= 1, \dots, N$ (and thus with $\mathcal{H}_N^E$) we denote by $\Gamma$. Further, we decompose $\mathfrak{H}_N^E$ into the orthogonal sum
\begin{equation}\label{representations}
\mathfrak{H}_N^E= \underset{\alpha\in A}{\oplus}\mathfrak{H}_N^{D_\alpha, E},
\end{equation}
where $\mathfrak{H}_N^{D_\alpha, E}$ consists of functions which transform under $O_\gamma$ according to some irreducible representation $D_\alpha$ of $\Gamma$, and $A$ is the set indexing all such irreducible representations. The decomposition \eqref{representations} reduces $\mathcal{H}_N^E$. We denote the selfadjoint restrictions of $\mathcal{H}_N^E$ to $\mathfrak{H}_N^{D_\alpha, E}$ by $\mathcal{H}_N^{D_\alpha, E}$. For any fixed irreducible representation $D$ with dimension $d_D$ and character $\zeta_D$ the orthogonal projector in $\mathfrak{H}_N$ onto the subspace of functions which transform according to $D$ is
\begin{equation*}\label{projection on D}
P^D:= d_D\int_\Gamma\overline{\zeta_D(\gamma)}O_\gamma d\mu(\gamma),
\end{equation*}
where $\mu$ is the invariant probability measure on $\Gamma$.

For $j= 1, 2$ let $D_j$ be some irreducible representations of $\Gamma$ with dimensions $d_{D_j}$ and characters $\zeta_{D_j}$. The corresponding projectors in $\mathfrak{H}_{Z, j}$ are given by
\begin{equation*}\label{projection on D_j}
P^{D_j}= d_{D_j}\int_\Gamma\overline{\zeta_{D_j}(\gamma)}O_{\gamma, j}d\mu(\gamma),
\end{equation*}
where $O_{\gamma, j}$ is the restriction of $O_\gamma$ to $\mathfrak H_{Z, j}$:
\begin{equation*}\label{subsystem rotation}
(O_{\gamma,j}\psi)(\x_{n_1}, \dots, \x_{n_{N_j}})= \prod_{n\in Z_j}e^{-i\varphi_\gamma\mathbf{n}_\gamma\cdot\mathbf{S}_n}\psi(\gamma^{-1}\x_{n_1}, \dots, \gamma^{-1}\x_{n_{N_j}}).
\end{equation*}
Given representations $D_j$ and $E_j$, projector $P^{D_j}P^{E_j}= P^{E_j}P^{D_j}$ reduces $\widetilde{\mathcal{H}}_{Z, j}$. We denote the reduced operators in 
\begin{equation*}\label{H^D_j^E_j}
\mathfrak{H}_{Z, j}^{D_j, E_j}:= P^{D_j}P^{E_j}\mathfrak{H}_{Z, j}
\end{equation*}
by $\widetilde{\mathcal{H}}_{Z, j}^{D_j, E_j}$. Let
\begin{equation}\label{E}
\varkappa_j(Z, D_j, E_j):= \inf \sigma(\widetilde{\mathcal{H}}_{Z, j}^{D_j, E_j}).
\end{equation}
We write $(D_1, E_1; D_2, E_2)\underset Z\prec (D, E)$ if the corresponding term cannot be omitted on the \rhs of
\begin{equation*}\label{symmetry embedding 2}
\mathfrak{H}_N^{D, E}\subset \underset{\substack{(D_1, E_1)\\ (D_2, E_2)}}{\oplus}\big(\mathfrak{H}_{Z, 1}^{D_1, E_1}\otimes\mathfrak{H}_{Z, 2}^{D_2, E_2}\big)
\end{equation*}
without violation of the inclusion.
For $Z_2\neq \varnothing$ let
\begin{equation}\label{bottom}\begin{split}
&\varkappa(Z, D, E)\\ &:= \!\begin{cases}\inf\!\big\{\varkappa_1(Z, D_1, E_1)\!+\! \varkappa_2(Z, D_2, E_2): (D_1, E_1; D_2, E_2)\underset Z\prec (D, E)\big\}, &\!\!\!\! Z_1\neq \varnothing,\\ \varkappa_2(Z, D, E),& \!\!\!\! Z_1= \varnothing.\end{cases}
\end{split}\end{equation}

The main result of the article is
\begin{theorem}\label{HVZ theorem}
For $N\in\mathbb N$ let $D$ be some irreducible representation of $\Gamma$, and $E$ some irreducible representation of $\Pi$, such that $P^DP^E\neq 0$. Then
\begin{equation*}\label{HVZ}
\sigma_{\mathrm{ess}}(\mathcal{H}_N^{D, E})= \big[\varkappa(D, E), \infty\big),
\end{equation*}
where
\begin{equation}\label{E(D)}
\varkappa(D, E)= \min\big\{\varkappa(Z, D, E): Z=(Z_1, Z_2), \: Z_2\neq \varnothing\big\}.
\end{equation}
\end{theorem}

\begin{remark}\label{weakening of lower assumption remark}
We only need Assumption~\ref{lower assumption} for the operators $\widetilde{\mathcal{H}}_{Z, j}^{D_j, E_j}$ which appear in \eqref{E}, \eqref{bottom}.
\end{remark}

\section{Commutator Estimates}\label{facts}

\subsection{One Particle Commutator Estimate}

\begin{lemma}\label{commutator lemma} 
Let $\chi\in C^2(\overline{\mathbb{R}^3})$ (i. e. a twice--differentiable function with bounded derivatives). Then for $m_n> 0$ the norm of the operator 
\[
[\chi, \Lambda_n]:L_2(\mathbb{R}^3, \mathbb{C}^4)\rightarrow H^1(\mathbb{R}^3, \mathbb{C}^4)
\]
satisfies
\begin{equation}\label{commutator}
\big\|[\chi, \Lambda_n]\big\|_{L_2(\mathbb{R}^3, \mathbb{C}^4)\rightarrow
H^1(\mathbb{R}^3, \mathbb{C}^4)}\leqslant C(m_n)\big(\|\nabla\chi\|_{L_\infty}+ \|\partial^2\chi\|_{L_\infty}\big). 
\end{equation}
Here $\|\partial^2\chi\|_{L_\infty}= \underset{\substack{\mathbf{z}\in\mathbb{R}^3\\ k,l\in\{1, 2,
    3\}}}{\max}\big|\partial_{kl}^2\chi(\mathbf{z})\big|$.
\end{lemma}

\begin{proof}{}
In the coordinate representation for $f\in C_0^1(\mathbb{R}^3, \mathbb{C}^4)$ the operator $\Lambda_n$ acts as
\begin{equation*}\begin{split}\label{Lambda_n in configuration space}
&(\Lambda_nf)(\x)= \frac{f(\x)}{2}+ \frac{im_n}{2\pi^2}\underset{\varepsilon\rightarrow +0}{\lim}\int\limits_{|\y- \x|\geqslant \varepsilon}\frac{\al\cdot(\x- \y)}{\rxy^3}K_1\big(m_n\rxy\big)f(\y)d\y\\ &+ \frac{m_n^2}{4\pi^2}\int\limits_{\mathbb{R}^3}\bigg(\beta\frac{K_1\big(m_n|\x- \y|\big)}{|\x- \y|}+ \frac{i\al\cdot(\mathbf{x}- \mathbf{y})}{|\x- \y|^2}K_0\big(m_n\rxy\big)\bigg)f(\y)d\y,
\end{split}\end{equation*}
where the limit on the \rhs is the limit in $L_2(\mathbb{R}^3, \mathbb{C}^4)$ (see Appendix~B of \cite{MorozovVugalter2006}, where this formula is derived in the case $m_n= 1$). The rest of the proof is an obvious modification of the proof of Lemma~1 of \cite{MorozovVugalter2006}, where the case $m_n= 1$ is considered.
\end{proof}

\begin{remark}\label{no mass dependence}
Since we only deal with a finite number of particles with positive masses, we will not trace the $m$-dependence of the constant in \eqref{commutator} any longer.
\end{remark}

\subsection{Many Particle Commutator Estimate}

\begin{lemma}\label{multiplicator lemma}
For any $d, k\in\mathbb{N}$ there exists $C> 0$ such that for any $\chi\in C^1(\overline{\mathbb{R}^d})$ and $u\in H^{1/2}(\mathbb{R}^d, \mathbb{C}^k)$
\begin{equation}\label{multiplier}
\|\chi u\|_{H^{1/2}(\mathbb{R}^d, \mathbb{C}^k)}\leqslant C\big(\|\chi\|_{L_\infty(\mathbb{R}^d)}+ \|\nabla\chi\|_{L_\infty(\mathbb{R}^d)}\big)\|u\|_{H^{1/2}(\mathbb{R}^d, \mathbb{C}^k)}.
\end{equation}
\end{lemma}
\begin{proof}{ of Lemma~\ref{multiplicator lemma}}
We choose the norm in $H^{1/2}(\mathbb{R}^d, \mathbb{C}^k)$ as (see \cite{Adams1975}, Theorem 7.48).
\begin{equation*}\label{the norm}
\|u\|_{H^{1/2}(\mathbb{R}^d, \mathbb{C}^k)}^2:= \|u\|_{L_2(\mathbb{R}^d, \mathbb{C}^k)}^2+ \iint\frac{\big|u(\x)- u(\y)\big|^2}{|\x- \y|^{d+ 1}}d\x d\y.
\end{equation*}
Then
\begin{equation}\label{chain}\begin{split}
\|\chi u\|_{H^{1/2}(\mathbb{R}^d, \mathbb{C}^k)}^2&= \|\chi u\|_{L_2(\mathbb{R}^d, \mathbb{C}^k)}^2+ \iint\frac{\big|\chi(\x)u(\x)- \chi(\y)u(\y)\big|^2}{|\x- \y|^{d+ 1}}d\x d\y\\ \leqslant \|\chi\|_{L_\infty}^2\|u\|_{L_2}^2&+ \iint\bigg(\frac{\big|\chi(\x)\big|^2\big|u(\x)- u(\y)\big|^2}{|\x- \y|^{d+ 1}}+ \frac{\big|\chi(\x)- \chi(\y)\big|^2\big|u(\y)\big|^2}{|\x- \y|^{d+ 1}}\bigg)d\x d\y\\ &\leqslant \|\chi\|_{L_\infty}^2\|u\|_{H^{1/2}}^2+ \underset{\y\in\mathbb{R}^d}{\sup}\int\frac{\big|\chi(\x)- \chi(\y)\big|^2}{|\x- \y|^{d+ 1}}d\x\|u\|_{L_2}^2.
\end{split}\end{equation}
The supremum on the \rhs of \eqref{chain} can be estimated as
\begin{equation}\label{estimate for integral}\begin{split}
&\underset{\y\in\mathbb{R}^d}{\sup}\int\frac{\big|\chi(\x)- \chi(\y)\big|^2}{|\x- \y|^{d+ 1}}d\x\leqslant \underset{\y\in\mathbb{R}^d}{\sup}\int_{|\x- \y|\leqslant 1}\frac{\big|\chi(\x)- \chi(\y)\big|^2}{|\x- \y|^{d+ 1}}d\x\\ &+ \underset{\y\in\mathbb{R}^d}{\sup}\int_{|\x- \y|> 1}\frac{\big|\chi(\x)- \chi(\y)\big|^2}{|\x- \y|^{d+ 1}}d\x\leqslant |\mathbb S^{d- 1}|\big(\|\nabla\chi\|_{L_\infty}^2+ 4\|\chi\|_{L_\infty}^2\big),
\end{split}\end{equation}
where $|\mathbb S^{d- 1}|$ is the area of $(d- 1)$--dimensional unit sphere. Substituting \eqref{estimate for integral} into \eqref{chain} we obtain \eqref{multiplier}.
\end{proof}

\begin{lemma}\label{boundedness of many-particle commutator in H^1/2}
For any $\chi\in C^2(\overline{\mathbb{R}^{3N}})$ the operator $[\chi, \Lambda^N]$ is bounded in\\ $H^{1/2}(\mathbb{R}^{3N}, \mathbb{C}^{4^N})$, and for any $\psi\in H^{1/2}(\mathbb{R}^{3N}, \mathbb{C}^{4^N})$ we have
\begin{equation}\label{jjj}\begin{split}
&\big\|[\chi, \Lambda^N]\psi\big\|_{H^{1/2}}\leqslant C\big(\|\nabla\chi\|_{L_\infty}+ \|\partial^2\chi\|_{L_\infty}\big)\big(\|\chi\|_{L_\infty}+ \|\nabla\chi\|_{L_\infty}\big)\|\psi\|_{H^{1/2}}
\end{split}\end{equation}
with $C$ depending only on $N$ and the masses of the particles.
\end{lemma}
\begin{proof}{}
Successively commuting $\chi$ with $\Lambda_n, \, n= 1, \dots, N$ (see \eqref{Lambda^N}) we obtain
\begin{equation}\label{reduction to one particle}
[\chi, \Lambda^N]= \sum_{n= 1}^N\prod_{k= 1}^{n- 1}\Lambda_k[\chi, \Lambda_n]\prod_{l= n+ 1}^N\Lambda_l,
\end{equation}
where the empty products should be replaced by identity operators.
By \eqref{Lambda_m} the operators $\Lambda_n$ are bounded in $H^{1/2}$ for any $n= 1, \dots, N$. This, together with \eqref{reduction to one particle}, and Lemmata~\ref{commutator lemma} and \ref{multiplicator lemma}, implies \eqref{jjj}. 
\end{proof}

\section{Lower Bound of the Essential Spectrum}\label{hard part}

In this section we prove that 
\begin{equation}\label{very hard claim}
\inf\sigma_{\mathrm{ess}}(\mathcal{H}_N^{D, E})\geqslant \varkappa(D, E).
\end{equation}

\subsection{Partition of Unity}

\begin{lemma}\label{partition lemma}
There exists a set of nonnegative functions $\{\chi_Z\}$ indexed by possible $2$--cluster decompositions $Z= (Z_1, Z_2)$ satisfying
\begin{eqnarray}
&1.& \chi_Z\in C^\infty(\mathbb{R}^{3N})\ \textrm{for all}\ Z;\notag\label{smooth cond}\\
&2.& \chi_Z(\kappa\X) = \chi_Z(\X)\ \textrm{for all}\ |\X| = 1,\ \kappa > 1,\ Z_2\neq \varnothing;\notag\label{homogen cond}\\
&3.& \sum_Z\chi_Z^2(\X)= 1,\ \textrm{for all}\ \X\in\mathbb{R}^{3N};\label{sum cond}\\
&4.& \begin{aligned}&\textrm{There exists}\ C > 0\ \textrm{such that for any}\ \X\in\textrm{\emph{supp}}~\chi_Z\\
&\min \big\{|\x_j - \x_n|: \x_j\in Z_1,\ \x_n\in Z_2;\ |\x_n|: \x_n\in Z_2\big\}> C|\X|;\label{separation cond}\end{aligned}\\
&5.& \begin{aligned}&\chi_Z(\gamma\x_1, \dots, \gamma\x_N)= \chi_Z(\x_1, \dots, \x_N)\ \textrm{for any orthogonal}\\ 
&\textrm{transformation}\ \gamma;\notag\label{orthoinvariance cond}\end{aligned}\\
&6.& \chi_Z \ \textrm{is invariant under permutations of variables preserving}\ Z_{1, 2}.\notag\label{permutation cond}
\end{eqnarray}
\end{lemma}

\begin{proof}{}
The proof is essentially based on the modification of the argument given in \cite{Simon1977CMP}, Lemma~2.4.

\paragraph{1.}
We first prove that for any $\X= (\x_1, \dots, \x_N)\in \mathbb{R}^{3N}$ with $|\X|= 1$ there exists a $2$--cluster decomposition $Z= (Z_1, Z_2)$ such that
\begin{equation}\label{good property}
\min \big\{|\x_j - \x_n|: \x_j\in Z_1,\ \x_n\in Z_2;\ |\x_n|: \x_n\in Z_2\big\}> N^{-3/2}.
\end{equation}
Indeed, let $k$ be such that $|\x_k|\geqslant |\x_j|$ for all $j= 1, \dots, N$. Then, since $|\X|= 1$,
\begin{equation}\label{notless}
|\x_k|\geqslant N^{-\frac{1}{2}}.
\end{equation}
Choose Cartesian coordinates in $\mathbb{R}^3$ with the first axis passing through the origin and $\x_k$, so that $\x_k= \big(|\x_k|, 0, 0\big)$. Consider $N$ regions
\begin{equation*}\begin{split}\label{slices}
R_1&:= \big\{\x\in \mathbb{R}^3: x^1\leqslant |\x_k|/N\big\},\\
R_l&:= \Big\{\x\in \mathbb{R}^3: x^1\in \big((l-1)|\x_k|/N,\ l|\x_k|/N\big]\Big\},\quad l= 2, \dots, N.
\end{split}\end{equation*}
At least one of these regions does not contain $\x_j$ with $j\neq k$. Let $l_0$ be the maximal index of such regions. Let $Z_2$ be the set of indices $n$ such that $\x_n\in\underset{l> l_0}{\cup}R_l$. $Z_2$ is nonempty since $\x_k\in Z_2$. Setting $Z_1:= I\setminus Z_2$ we observe that
\begin{equation*}
\min \big\{|\x_j - \x_n|: \x_j\in Z_1,\ \x_n\in Z_2;\ |\x_n|: \x_n\in Z_2\big\}> |\x_k|/N,
\end{equation*}
which together with \eqref{notless} implies \eqref{good property}.

\paragraph{2.}
Choose $\eta\in C^\infty(\mathbb R_+)$ so that
\begin{equation*}
\eta(t)\equiv \begin{cases}0, & t\in [0, 1]\\ 1, & t\in [2, \infty).\end{cases}
\end{equation*}
Let
\begin{equation}\label{zeta}
\zeta_Z(\X):= \begin{cases}\displaystyle\prod_{\substack{j\in Z_1\\ n\in Z_2}}\eta\bigg(\frac{2|\x_j- \x_n|}{|\X|N^{-3/2}}\bigg)\prod_{n\in Z_2}\eta\bigg(\frac{2|\x_n|}{|\X|N^{-3/2}}\bigg)\eta\big(2|\X|\big), & Z_2\neq \varnothing\\ 1- \eta\big(2|\X|\big), & Z_2= \varnothing.\end{cases}
\end{equation}
Functions \eqref{zeta} satisfy conditions 1, 2, 4 (with $C= N^{-3/2}$), 5, and 6 of Lemma~\ref{partition lemma}. Moreover, by the first part of the proof
\begin{equation*}
\sum_Z\zeta_Z(\X)\geqslant 1, \quad \textrm{for all}\quad \X\in \mathbb R^{3N}.
\end{equation*}
Hence all the conditions are satisfied by the functions
\begin{equation*}
\chi_Z:= \zeta_Z^{1/2}\Big(\sum_Z\zeta_Z\Big)^{-1/2}.
\end{equation*}
\end{proof}

Let
\begin{equation}\label{chi_Z^R}
\chi_Z^R(\X):= \chi_Z(\X/R),
\end{equation}
where the functions $\chi_Z$ are defined in Lemma~\ref{partition lemma}.
The derivatives of $\chi_Z^R$ decay as $R$ tends to infinity:
\begin{equation}\label{derivatives decay}
\|\nabla\chi_Z^R\|_\infty\leqslant CR^{-1}, \quad \|\partial^2\chi_Z^R\|_\infty\leqslant CR^{-2}.
\end{equation}
To simplify the notation we omit the superscript $R$ further on.

\subsection{Cluster Decomposition and Lower Bound}

We now estimate from below the quadratic form of $\mathcal{H}_N^{D, E}$ on a function $\psi$ from $\mathfrak{H}_N^{D, E}\cap\Lambda^N\underset{n= 1}{\overset{N}{\otimes}}H^{1/2}(\mathbb{R}^3, \mathbb{C}^4)$, which is the form domain of $\mathcal{H}_N^{D, E}$.
\begin{equation*}\label{expansion 1}\begin{split}
\langle\mathcal{H}_N^{D, E}\psi, \psi\rangle&= \langle\bigg(\sum_{n= 1}^N(D_n+ V_n)+ \sum_{n<j}^NU_{nj}\bigg)\sum_Z\chi_Z^2\psi, \psi\rangle\\ &= \sum_Z\langle\bigg(\sum_{n= 1}^N(D_n+ V_n)+ \sum_{n<j}^NU_{nj}\bigg)\chi_Z\psi, \chi_Z\psi\rangle.
\end{split}\end{equation*}
Here we have used \eqref{sum cond} and the relation
\begin{equation}\label{first order commutation}
\sum_Z\langle f, \sum_{n= 1}^N\nabla_n(\chi_Z^2 g)\rangle= \sum_Z\langle\chi_Z f, \sum_{n= 1}^N\nabla_n(\chi_Z g)\rangle+ \sum_Z\langle f, \sum_{n= 1}^N\nabla_n\Big(\frac{\chi_Z^2}{2}\Big)g\rangle
\end{equation}
which holds for any $f, g\in \underset{n= 1}{\overset{N}{\otimes}}H^{1/2}(\mathbb{R}^3, \mathbb{C}^4)$. The last term on the \rhs of \eqref{first order commutation} is equal to zero due to \eqref{sum cond}. Thus
\begin{equation}\label{expansion 2}\begin{split}
\langle\mathcal{H}_N^{D, E}\psi, \psi\rangle&= \sum_{Z= (Z_1, Z_2)}\bigg(\langle(\mathcal{H}_{Z, 1}+ \mathcal{H}_{Z, 2})\Lambda^N\chi_Z\psi, \Lambda^N\chi_Z\psi\rangle\\&+ \langle(\mathcal{H}_{Z, 1}+ \mathcal{H}_{Z, 2})[\chi_Z, \Lambda^N]\psi, \Lambda^N\chi_Z\psi\rangle\\&+ \langle(\mathcal{H}_{Z, 1}+ \mathcal{H}_{Z, 2})\chi_Z\psi, [\chi_Z, \Lambda^N]\psi\rangle\\&+ \langle\sum_{n\in Z_2}V_n\chi_Z^2\psi, \psi\rangle+ \langle\sum_{\substack{n< j\\ n\# j}}U_{nj}\chi_Z^2\psi, \psi\rangle\bigg).
\end{split}\end{equation}
The terms at the last line of \eqref{expansion 2} can be estimated as
\begin{equation}\label{last line terms}
\langle\sum_{n\in Z_2}V_n\chi_Z^2\psi, \psi\rangle+ \langle\sum_{\substack{n< j\\ n\# j}}U_{nj}\chi_Z^2\psi, \psi\rangle\geqslant -\varepsilon_1(R)\big(\langle\mathcal{H}_N^{D, E}\psi, \psi\rangle+ \|\psi\|^2\big)
\end{equation}
with $\varepsilon_1(R)\to 0$ as $R\to \infty$, due to \eqref{V_n decay}, \eqref{U_nj decay}, \eqref{separation cond}, \eqref{chi_Z^R}, and \eqref{backward equivalence}. The terms at the second and third lines of \eqref{expansion 2} can also be estimated as
\begin{equation*}\label{2,3 line terms}\begin{split}
\langle(\mathcal{H}_{Z, 1}+ \mathcal{H}_{Z, 2})[\chi_Z, \Lambda^N]\psi, \Lambda^N\chi_Z\psi\rangle+ \langle(\mathcal{H}_{Z, 1}+ \mathcal{H}_{Z, 2})\chi_Z\psi, [\chi_Z, \Lambda^N]\psi\rangle\\ \geqslant -\varepsilon_2(R)\big(\langle\mathcal{H}_N^{D, E}\psi, \psi\rangle+ \|\psi\|^2\big), \quad \varepsilon_2(R)\underset{R\to \infty}{\longrightarrow} 0,
\end{split}\end{equation*}
due to \eqref{upper estimate}, \eqref{multiplier}, \eqref{jjj}, \eqref{derivatives decay}, and \eqref{backward equivalence}.
In the case $Z_2\neq \varnothing$ we estimate the terms at the first line of \eqref{expansion 2} in the following way (recall the definitions of $\varkappa_j(Z, D_j, E_j)$ and $\varkappa(D, E)$ given in \eqref{E} and \eqref{E(D)}):
\begin{equation}\label{main terms}\begin{split}
&\langle(\mathcal{H}_{Z, 1}+ \mathcal{H}_{Z, 2})\Lambda^N\chi_Z\psi, \Lambda^N\chi_Z\psi\rangle\\& = \sum_{(D_1, E_1; D_2, E_2)\prec (D, E)}\langle(\mathcal{H}_{Z, 1}P^{D_1}P^{E_1}+ \mathcal{H}_{Z, 2}P^{D_2}P^{E_2})\Lambda^N\chi_Z\psi, \Lambda^N\chi_Z\psi\rangle\\ &\geqslant \sum_{(D_1, E_1; D_2, E_2)\prec (D, E)}\langle\big(\varkappa_1(Z, D_1, E_1)P^{D_1}P^{E_1}\\ &+ \varkappa_2(Z, D_2, E_2)P^{D_2}P^{E_2}\big)\Lambda^N\chi_Z\psi, \Lambda^N\chi_Z\psi\rangle\\ &\geqslant \varkappa(D, E)\langle\Lambda^N\chi_Z\psi, \Lambda^N\chi_Z\psi\rangle\\ &= \varkappa(D, E)\langle\chi_Z^2\psi, \psi\rangle+ \varkappa(D, E)\langle[\Lambda^N, \chi_Z]\psi, \chi_Z\psi\rangle\\ &+ \varkappa(D, E)\langle\Lambda^N\chi_Z\psi, [\Lambda^N, \chi_Z]\psi\rangle.
\end{split}\end{equation}
By \eqref{multiplier}, \eqref{jjj}, \eqref{derivatives decay}, and \eqref{backward equivalence} the last two terms on the \rhs of \eqref{main terms} can be estimated as
\begin{equation}\label{so}\begin{split}
\varkappa(D, E)\langle[\Lambda^N, \chi_Z]\psi, \chi_Z\psi\rangle+ \varkappa(D, E)\langle\Lambda^N\chi_Z\psi, [\Lambda^N, \chi_Z]\psi\rangle\\ \geqslant -\varepsilon_3(R)\big(\langle\mathcal{H}_N^{D, E}\psi, \psi\rangle+ \|\psi\|^2\big),\quad \varepsilon_3(R)\underset{R\to \infty}{\longrightarrow}0.
\end{split}\end{equation}
Substituting the estimates \eqref{last line terms} --- \eqref{so} into \eqref{expansion 2} we obtain
\begin{equation}\begin{split}\label{out of compact}
\langle\mathcal{H}_N^{D, E}\psi, \psi\rangle&\geqslant \varkappa(D, E)\langle\sum_{\substack{Z= (Z_1, Z_2)\\ Z_2\neq \varnothing}}\chi_Z^2\psi, \psi\rangle+ \langle\mathcal{H}_N^{D, E}\Lambda^N\chi_{(I, \varnothing)}\psi, \Lambda^N\chi_{(I, \varnothing)}\psi\rangle\\ &- \varepsilon_4(R)\big(\langle\mathcal{H}_N^{D, E}\psi, \psi\rangle+ \|\psi\|^2\big), \quad \varepsilon_4(R)\underset{R\to \infty}{\longrightarrow}0.
\end{split}\end{equation}

\subsection{Estimate Inside of the Compact Region}

It remains to estimate from below the quadratic form of the operator $\mathcal{H}_N^{D, E}$ on the function $\Lambda^N\chi_{(I, \varnothing)}\psi$. Note that according to Lemma~\ref{partition lemma} and \eqref{chi_Z^R} $\supp\chi_{(I, \varnothing)}\subset [-2R, 2R]^{3N}$. To simplify the notation let 
\begin{equation*}\label{chi_0}
\chi_0:= \chi_{(I, \varnothing)}.
\end{equation*}

\begin{lemma}\label{lemma2}
For $M> 0$ let
\[
W_M:=\big\{\mathbf{p}\in\mathbb{R}^{3N}: |p_i|\leqslant M, i=1,\dots,3N\big\}, \:\widetilde W_M:=\mathbb{R}^{3N}\setminus W_M.
\]
There exists a finite set $Q_M\subset L_2(\mathbb{R}^{3N})$ such that for any $f\in L_2(\mathbb{R}^{3N})$ with $\mathrm{supp}f\subset[-2R, 2R]^{3N},\: f\bot Q_M$ holds
\begin{equation*}\label{much is away}
\|\hat f\|_{L_2(\widetilde W_M)}\geqslant \frac{1}{2}\|\hat f\|_{L_2(\mathbb{R}^{3N})}. 
\end{equation*}
\end{lemma}

The proof of Lemma~\ref{lemma2} is analogous to the proof of Theorem~7 of \cite{VugalterWeidl2003} and is given in Appendix~C of \cite{MorozovVugalter2006}.

\medskip

It follows from \eqref{assumption 1} that for any $M> 0$
\begin{equation}\begin{split}\label{step1}
\langle \mathcal{H}_N^{D, E}\Lambda^N\chi_0\psi, \Lambda^N\chi_0\psi\rangle\geqslant C_1\langle\underset{n= 1}{\overset{N}{\sum}}D_n I_{\widetilde W_M}\Lambda^N\chi_0\psi, \Lambda^N\chi_0\psi\rangle- C_2\|\chi_0\psi\|^2.
\end{split}\end{equation}
Here $I_{\widetilde W_M}$ is the operator of multiplication by the characteristic function of $\widetilde W_M$ in momentum space.

We choose
\begin{equation}\label{M}
M:= 8\big(\varkappa(D, E)+ C_2\big)C_1^{-1}
\end{equation}
and assume henceforth that $f:= \chi_0\psi$ is orthogonal to the set $Q_M$ defined in Lemma~\ref{lemma2}. Since in momentum space the operator $D_n$ acts on functions from $\Lambda_nL_2(\mathbb{R}^3, \mathbb{C}^4)$ as multiplication by $\sqrt{|\mathbf{p}|^2+ m_n^2}$, by construction of $\widetilde W_M$ we have
\begin{equation}\label{big kinetic energy}
\langle\underset{n= 1}{\overset{N}{\sum}}D_n I_{\widetilde W_M}\Lambda^N\chi_0\psi, \Lambda^N\chi_0\psi\rangle\geqslant M\|I_{\widetilde W_M}\Lambda^N\chi_0\psi\|^2.
\end{equation}
Inequalities \eqref{step1} and \eqref{big kinetic energy} imply
\begin{equation}\begin{split}\label{step2}
&\langle\mathcal{H}_N^{D, E}\Lambda^N\chi_0\psi, \Lambda^N\chi_0\psi\rangle\geqslant C_1M\|I_{\widetilde W_M}\Lambda^N\chi_0\psi\|^2- C_2\|\chi_0\psi\|^2\\ &\geqslant C_1M\Big(\|I_{\widetilde W_M}\chi_0\psi\|- \big\|I_{\widetilde W_M}[\Lambda^N, \chi_0]\psi\big\|\Big)^2- C_2\|\chi_0\psi\|^2\\ &\geqslant C_1M\Big(\frac{1}{2}\|I_{\widetilde W_M}\chi_0\psi\|^2- \big\|I_{\widetilde W_M}[\Lambda^N, \chi_0]\psi\big\|^2\Big)- C_2\|\chi_0\psi\|^2\\ &\geqslant 4\big(\varkappa(D, E)+ C_2\big)\|I_{\widetilde W_M}\chi_0\psi\|^2\\ &- 8\big(\varkappa(D, E)+ C_2\big)\big\|[\Lambda^N, \chi_0]\psi\big\|^2- C_2\|\chi_0\psi\|^2.
\end{split}\end{equation}
At the last step we have used \eqref{M}. 
The second term on the \rhs of \eqref{step2} can be estimated analogously to \eqref{so} as
\begin{equation*}
- 8\big(\varkappa(D, E)+ C_2\big)\big\|[\Lambda^N, \chi_0]\psi\big\|^2\geqslant -\varepsilon_5(R)\big(\langle\mathcal{H}_N^{D, E}\psi, \psi\rangle+ \|\psi\|^2\big),\ \varepsilon_5(R)\underset{R\to \infty}{\longrightarrow}0.
\end{equation*}
For the first term on the \rhs of \eqref{step2} Lemma~\ref{lemma2} implies
\begin{equation}\label{big away}
4\|I_{\widetilde W_M}\chi_0\psi\|^2\geqslant \|\chi_0\psi\|^2.
\end{equation}
As a consequence of \eqref{step2} --- \eqref{big away}, we have
\begin{equation}\label{another half}\begin{split}
\langle\mathcal{H}_N^{D, E}\Lambda^N\chi_0\psi, \Lambda^N\chi_0\psi\rangle&\geqslant \varkappa(D, E)\|\chi_0\psi\|^2- \varepsilon_5(R)\big(\langle\mathcal{H}_N^{D, E}\psi, \psi\rangle+ \|\psi\|^2\big),\\ &\varepsilon_5(R)\underset{R\to \infty}{\longrightarrow}0.
\end{split}\end{equation}

\subsection{Completion of the Proof}\label{completion}

By \eqref{out of compact}, \eqref{another half}, and \eqref{sum cond}
\begin{equation*}\label{final}
\langle \mathcal{H}_N^{D, E}\psi, \psi\rangle\geqslant \varkappa(D, E)\|\psi\|^2- \varepsilon_6(R)\big(\langle\mathcal{H}_N^{D, E}\psi, \psi\rangle+ \|\psi\|^2\big),\quad \varepsilon_6(R)\underset{R\to \infty}{\longrightarrow}0.
\end{equation*}
for any $\psi$ in the form domain of $\mathcal{H}_N^{D, E}$ orthogonal to the finite set of functions (cardinality of this set depends on $R$). This implies the discreteness of the spectrum of $\mathcal{H}_N^{D, E}$ below $\varkappa(D, E)$ and thus \eqref{very hard claim}.

\section{Spectrum of the Free Cluster}\label{easy part}

In this section we characterize the spectrum of the cluster $Z_2$ which does not interact with the external field.

\begin{proposition}\label{reduction for H^D,E}
For any irreducible representations $D_2, E_2$ of rotation--reflection and permutation groups the spectrum of $\widetilde{\mathcal{H}}_{Z, 2}^{D_2, E_2}$ is
\begin{equation*}\label{free subsystem spectrum with symmetries}
\sigma(\widetilde{\mathcal{H}}_{Z, 2}^{D_2, E_2})= \sigma_{\textrm{\emph{ess}}}(\widetilde{\mathcal{H}}_{Z, 2}^{D_2, E_2})= \big[\varkappa_2(Z, D_2, E_2), \infty\big),
\end{equation*}
with some $\varkappa_2(Z, D_2, E_2)\in\mathbb R$.
\end{proposition}

\begin{proof}{}
Let us introduce the new coordinates in the configuration space $\mathbb{R}^{3N_2}$ of the cluster $Z_2= \{n_1, \dots, n_{N_2}\}$, in the same manner as it is done in \cite{LewisSiedentopVugalter1997}. Let $M:= \sum_{n\in Z_2}m_n$ be the total mass of the particles constituting the cluster. We introduce
\begin{equation}\label{new coordinates Z_2}\begin{split}
\y_0&:= \frac{1}{M}\sum_{n\in Z_2}m_n\x_n,\\
\y_k&:= \x_{n_{k+ 1}}- \x_{n_1}, \quad k= 1, \dots, N_2- 1.
\end{split}\end{equation}
The Jacobian of this variable change is one. $\y_0$ is the coordinate of the center of mass, whereby $\y_k$, $k= 1, \dots, N_2- 1$ are the internal coordinates of the cluster. 
Accordingly,
\begin{equation}\label{old coordinates}\begin{split}
\x_{n_1}&= \y_0- \frac1M\sum_{k= 1}^{N_2- 1}m_{n_{k+ 1}}\y_k,\\
\x_{n_{l+ 1}}&= \y_0+ \y_l- \frac1M\sum_{k= 1}^{N_2- 1}m_{n_{k+ 1}}\y_k, \quad l= 1, \dots, N_2- 1.
\end{split}\end{equation}
The momentum operators in the new coordinates are
\begin{equation}\label{momenta Z_2}\begin{split}
\mathbf p_{n_1}:= -i\nabla_{\x_{n_1}}&= \frac{m_{n_1}}{M}\p- \sum_{k= 1}^{N_2- 1}(-i\nabla_{\y_k}),\\ \mathbf p_{n_k}:= -i\nabla_{\x_{n_k}}&= \frac{m_{n_k}}{M}\p+ (-i\nabla_{\y_{k- 1}}), \quad k= 2, \dots, N_2,
\end{split}\end{equation}
where $\p$ is the total momentum of the cluster:
\begin{equation*}\label{new total momentum Z_2}
\p:= \sum_{n\in Z_2}-i\nabla_{\x_n}= -i\nabla_{\y_0}.
\end{equation*}
Let $\mathcal F_0$ be the partial Forurier transform on $\mathfrak{H}_{Z, 2}^{D_2, E_2}$ defined by
\begin{equation*}\label{F_0}
(\mathcal F_0f)(\p, \y_1, \dots, \y_{N_2- 1}):= \frac{1}{(2\pi)^{3/2}}\int_{\mathbb{R}^3} f(\y_0, \y_1, \dots, \y_{N_2- 1})e^{-i\p\y_0}d\y_0.
\end{equation*}
By \eqref{H_Z_2} we have
\begin{equation*}\label{operator for Z_2}
\widetilde{\mathcal{H}}_{Z, 2}^{D_2, E_2}= \mathcal F_0^{-1}\widehat\Lambda_{Z, 2}\widehat{\mathcal{H}}_{Z, 2}^{D_2, E_2}\widehat\Lambda_{Z, 2}\mathcal F_0,
\end{equation*}
where in the new coordinates
\begin{equation}\label{hat H on Z_2}
\widehat{\mathcal{H}}_{Z, 2}^{D_2, E_2}:= \sum_{n\in Z_2}(\al_n\cdot \mathbf{p}_n+ \beta_n m_n)+ \sum_{k= 2}^{N_2- 1}U_{n_1n_k}(\y_k)+ \sum_{1< k< l\leqslant N_2- 1}U_{n_kn_l}(\y_k- \y_l),
\end{equation}
\begin{equation}\label{Lambda_Z_2}
\widehat\Lambda_{Z, 2}:= \prod_{n\in Z_2}\widehat\Lambda_n,
\end{equation}
\begin{equation*}\label{Lambda p}
\widehat\Lambda_n:= \frac{1}{2}+ \frac{\al_n\cdot \mathbf{p}_n+ \beta_n m_n}{2\sqrt{\mathbf{p}_n^2+ m_n^2}},
\end{equation*}
$\mathbf p_n$ are givn by \eqref{momenta Z_2}, and $\p$ should now be interpreted as multiplication by the vector--function.
The operators \eqref{hat H on Z_2} and \eqref{Lambda_Z_2} obviously commute with $\mathfrak P:= |\p|$. The operator $\mathcal F_0^{-1}\mathfrak P \mathcal F_0$ (unlike $\mathcal F_0^{-1}\p\mathcal F_0$) is well--defined in $\mathfrak H_{Z, 2}^{D_2, E_2}$, since it commutes with $P^{D_2}$ an $P^{E_2}$ in $\mathfrak H_{Z, 2}$. This implies that $\widetilde{\mathcal{H}}_{Z, 2}^{D_2, E_2}$ commutes with $\mathcal F_0^{-1}\mathfrak P\mathcal F_0$.

Let $\omega:= \p/\mathfrak P\in S^2$.
We decompose the Hilbert space $\mathfrak{H}_{Z, 2}^{D_2, E_2}$ into the direct integral
\begin{equation}\label{direct decomposition-space Z_2}
\mathfrak{H}_{Z,2}^{D_2, E_2}= \int_0^\infty\oplus\mathfrak{H}_{Z, 2}^{D_2, E_2, \mathfrak{P}}\mathfrak P^2d\mathfrak{P}.
\end{equation}
The fibre space $\mathfrak{H}_{Z, 2}^{D_2, E_2, \mathfrak{P}}$ can be considered as a subspace of $L_2(\mathbb{R}^{3N_2- 3}\times S^2, \mathbb{C}^{4^{N_2}})$ with the inner product
\begin{equation*}\label{*-product}
\langle f,g\rangle_*:= \int_{\mathbb{R}^{3(N_2- 1)}\times S^2}\langle f,g\rangle_{\mathbb{C}^{4^{N_2}}}d\y_1\cdots d\y_{N_2- 1}d\omega.
\end{equation*}
For $f\in \mathfrak{H}_{Z, 2}^{D_2, E_2}$ the corresponding element of $\mathfrak{H}_{Z, 2}^{D_2, E_2, \mathfrak P}$ is given by
\begin{equation*}\label{reduction}
f_{\mathfrak P}:= \mathcal F_0f\arrowvert_{|\p|= \mathfrak P}.
\end{equation*}
We have
\begin{equation}\label{*-isometry}
\|f\|^2= \int_0^\infty\|f_{\mathfrak{P}}\|_*^2\mathfrak P^2d\mathfrak{P}
\end{equation}
in compliance with \eqref{direct decomposition-space Z_2}.
The form domain of $\widetilde{\mathcal{H}}_{Z, 2}^{D_2, E_2, \mathfrak{P}}$ is
\begin{equation*}
\mathfrak{D}^{\mathfrak{P}}:= \Lambda_{Z, 2}^{\mathfrak{P}}P^{D_2}P^{E_2}H^{1/2}(\mathbb{R}^{3(N_2- 1)}\times S^2, \mathbb{C}^{4^{N_2}}),
\end{equation*}
where $\Lambda_{Z, 2}^{\mathfrak{P}}$ is given by \eqref{Lambda_Z_2} with the only difference that we should replace $\p$ by $\omega\mathfrak P$ in \eqref{momenta Z_2}. The operators on fibres of the direct integral \eqref{direct decomposition-space Z_2} are
\begin{equation*}\label{operator for Z_2 on fibre}
\widetilde{\mathcal{H}}_{Z, 2}^{D_2, E_2, \mathfrak P}:= \Lambda_{Z, 2}^{\mathfrak{P}}\mathcal{H}_{Z, 2}^{D_2, E_2, \mathfrak{P}}\Lambda_{Z, 2}^{\mathfrak{P}},
\end{equation*}
where $\mathcal{H}_{Z, 2}^{D_2, E_2, \mathfrak{P}}$ is given by the \rhs of \eqref{hat H on Z_2} with $\p$ replaced by $\omega\mathfrak P$ in \eqref{momenta Z_2}. We thus have
\begin{equation}\label{direct integral decomposition}
\widetilde{\mathcal{H}}_{Z, 2}^{D_2, E_2}= \int_0^\infty\oplus\widetilde{\mathcal{H}}_{Z, 2}^{D_2, E_2, \mathfrak P}\mathfrak P^2d\mathfrak P.
\end{equation}

The spectrum of $\widetilde{\mathcal{H}}_{Z,2}^{D_2, E_2}$ can be represented as
\begin{equation}\label{spectrum decomposition Z_2}
\sigma(\widetilde{\mathcal{H}}_{Z, 2}^{D_2, E_2})= \overline{\textrm{ess}\bigcup_{\mathfrak{P}\in\mathbb R_+}\sigma(\widetilde{\mathcal{H}}_{Z, 2}^{D_2, E_2, \mathfrak{P}})},
\end{equation}
where the essential union is taken with respect to the Lebesgue measure in $\mathbb R_+$. The bottom of the spectrum of $\widetilde{\mathcal{H}}_{Z, 2}^{D_2, E_2, \mathfrak{P}}$ is given by
\begin{equation}\label{P-bottom Z_2}
\mu(\mathfrak{P}):= \underset{\psi\in \mathfrak{D}^{\mathfrak{P}}}{\inf}\frac{\langle\widetilde{\mathcal{H}}_{Z, 2}^{D_2, E_2, \mathfrak{P}}\psi, \psi\rangle_*}{\|\psi\|_*^2}.
\end{equation}

\begin{lemma}\label{mu continuity lemma 2}
Function \eqref{P-bottom Z_2} is continuous on $\mathbb R_+$.
\end{lemma}

\begin{proof}{ of Lemma~\ref{mu continuity lemma 2}}
Let us fix $\mathfrak{P}\in\mathbb{R}_+$ and $\varepsilon> 0$. We will prove that $\big|\mu(\mathfrak{P}+ \mathfrak{p})- \mu(\mathfrak{P})\big|< \varepsilon$ if $|\mathfrak{p}|$ is small enough. Choose $\psi\in \mathfrak{D}^{\mathfrak{P}}$ such that
\begin{equation}\label{close to mu 2}
\bigg|\frac{\langle\widetilde{\mathcal{H}}_{Z, 2}^{D_2, E_2, \mathfrak{P}}\psi, \psi\rangle_*}{\|\psi\|_*^2}- \mu(\mathfrak{P})\bigg|\leqslant \frac{\varepsilon}{2}.
\end{equation}
Let
\begin{equation*}\label{phi^P+p 2}
\phi:= \Lambda_{Z, 2}^{\mathfrak{P}+ \mathfrak{p}}\psi\in \mathfrak{D}^{\mathfrak{P}+ \mathfrak{p}}.
\end{equation*}
We have
\begin{equation}\label{phi-psi 2}
\phi- \psi= (\Lambda_{Z, 2}^{\mathfrak{P}+ \mathfrak{p}}- \Lambda_{Z, 2}^{\mathfrak{P}})\psi= \sum_{k= 1}^{N_2}\prod_{i< k}\Lambda_{n_i}^{\mathfrak{P}+ \mathfrak{p}}(\Lambda_{n_k}^{\mathfrak{P}+ \mathfrak{p}}- \Lambda_{n_k}^{\mathfrak{P}})\prod_{j> k}\Lambda_{n_j}^{\mathfrak{P}}\psi.
\end{equation}
Let $\mathcal{F}$ be the unitary Fourier transform in $L_2(\mathbb{R}^{3(N_2- 1)}\times S^2, \mathbb{C}^{4^{N_2}})$ defined by
\begin{equation*}\label{Fourier transform 2}\begin{split}
&(\mathcal{F}\xi)(\omega, \mathbf{q}_1, \dots, \mathbf{q}_{N_2- 1})\\ &:= (2\pi)^{3(1- N_2)/2}\int\limits_{\mathbb{R}^{3(N_2- 1)}}\xi(\omega, \y_1, \dots, \y_{N_2- 1})e^{-i\sum\limits_{k= 1}^{N_2-1}\mathbf{q}_k\cdot\y_k}d\y_1\cdots d\y_{N_2- 1}.
\end{split}\end{equation*}
We can rewrite \eqref{phi-psi 2} as
\begin{equation}\label{phi-psi Fourier 2}
\phi- \psi= \mathcal{F}^{-1}\sum_{k= 1}^{N_2}\prod_{i< k}\widehat\Lambda_{n_i}^{\mathfrak{P}+ \mathfrak{p}}(\widehat\Lambda_{n_k}^{\mathfrak{P}+ \mathfrak{p}}- \widehat\Lambda_{n_k}^{\mathfrak{P}})\prod_{j> k}\widehat\Lambda_{n_j}^{\mathfrak{P}}\mathcal{F}\psi,
\end{equation}
where $\widehat\Lambda_{n}^{\mathfrak{P}},\: n\in Z_2$ are the operators of multiplication by the symbols
\begin{equation}\label{Lambda^P Fourier 2}
\widehat\Lambda_n^{\mathfrak{P}}:= \frac{1}{2}+ \frac{\al_n\cdot \widehat{\mathbf{p}}_n+ \beta_n m_n}{2\sqrt{\widehat{\mathbf{p}}_n^2+ m_n^2}},
\end{equation}
\begin{equation}\label{p_n Fourier 2}\begin{split}
\widehat{\mathbf{p}}_{n_1}&:= \frac{m_{n_1}}{M}\omega\mathfrak{P}- \sum_{k= 1}^{N_2- 1}\mathbf{q}_k,\\ \widehat{\mathbf{p}}_{n_k}&:= \frac{m_{n_k}}{M}\omega\mathfrak{P}+ \mathbf{q}_{k- 1}, \quad k= 2, \dots, N_2.
\end{split}\end{equation}
The matrix--functions \eqref{Lambda^P Fourier 2} are uniformly continuous in $\mathfrak{P}$. Thus by \eqref{phi-psi Fourier 2}
\begin{equation}\label{closeness 2}
\|\phi- \psi\|_{H^{1/2}(\mathbb{R}^{3(N_2- 1)}\times S^2, \mathbb{C}^{4^{N_2}})}\leqslant C\sum_{k= 1}^{N_2}\|\Lambda_{n_k}^{\mathfrak{P}+ \mathfrak{p}}- \Lambda_{n_k}^{\mathfrak{P}}\|_{H^{1/2}\to H^{1/2}}\underset{|\mathfrak{p}|\to 0}{\longrightarrow}0.
\end{equation}
We write
\begin{equation}\label{star expansion 2}\begin{split}
&\langle\widetilde{\mathcal{H}}_{Z, 2}^{D_2, E_2, \mathfrak{P}+ \mathfrak{p}}\phi, \phi\rangle_*= \langle\widetilde{\mathcal{H}}_{Z, 2}^{D_2, E_2, \mathfrak{P}}\psi, \psi\rangle_*+  \langle\mathcal{H}_{Z, 2}^{D_2, E_2, \mathfrak{P}}(\phi- \psi), \psi\rangle_*\\ &+ \langle\mathcal{H}_{Z, 2}^{D_2, E_2, \mathfrak{P}}\phi, (\phi- \psi)\rangle_*+  \langle(\widetilde{\mathcal{H}}_{Z, 2}^{D_2, E_2, \mathfrak{P}+ \mathfrak{p}}- \mathcal{H}_{Z, 2}^{D_2, E_2, \mathfrak{P}})\phi, \phi\rangle_*.
\end{split}\end{equation}
The second and third terms on the \rhs of \eqref{star expansion 2} tend to zero as $|\mathfrak{p}|\to 0$ according to \eqref{closeness 2} and \eqref{upper estimate}. The last term also tends to zero for small $|\mathbf p|$, since the symbol of the difference is
\begin{equation*}\label{difference symbol 2}
\mathcal{F}(\widetilde{\mathcal{H}}_{Z, 2}^{D_2, E_2, \mathfrak{P}+ \mathfrak{p}}- \mathcal{H}_{Z, 2}^{D_2, E_2, \mathfrak{P}})\mathcal{F}^{-1}= \sum_{n\in Z_2}\frac{m_n}{M}\al_n\cdot\omega\mathfrak{p}.
\end{equation*}
From \eqref{closeness 2} and \eqref{star expansion 2} follows that
\begin{equation}\label{eps/2 2}
\bigg|\frac{\langle\widetilde{\mathcal{H}}_{Z, 2}^{D_2, E_2, \mathfrak{P}}\psi, \psi\rangle_*}{\|\psi\|_*^2}- \frac{\langle\widetilde{\mathcal{H}}_{Z, 2}^{D_2, E_2, \mathfrak{P}+ \mathfrak{p}}\phi, \phi\rangle_*}{\|\phi\|_*^2}\bigg|\leqslant \frac{\varepsilon}{2},
\end{equation}
if $|\mathfrak{p}|$ is small enough. Hence by \eqref{close to mu 2} and \eqref{eps/2 2} for any $\varepsilon> 0$
\[
\big|\mu(\mathfrak{P}+ \mathfrak{p})- \mu(\mathfrak{P})\big|< \varepsilon
\]
for $|\mathfrak p|$ small enough.
\end{proof}

Now we prove that $\mu$ is semibounded from below and tends to infinity as $|\mathfrak{P}|\to \infty$. This, together with \eqref{spectrum decomposition Z_2} and Lemma~\ref{mu continuity lemma 2}, implies that the spectrum of $\widetilde{\mathcal{H}}_{Z, 2}^{D_2, E_2}$ is purely essential and is concentrated on a semi--axis. Proposition~\ref{reduction for H^D,E} will be thus proved.

According to \eqref{assumption 1} for $j= 2$ and \eqref{square root formula} we have
\begin{equation}\label{with root 2}\begin{split}
\langle\widetilde{\mathcal{H}}_{Z, 2}^{D_2, E_2}\psi, \psi\rangle\geqslant C_1\langle\sum_{n\in Z_2}\sqrt{-\Delta_n + m_n^2}\psi, \psi\rangle- C_2\|\psi\|^2,\\ \textrm{for any}\quad \psi\in P^DP^E\underset{n\in Z_2}{\otimes}\Lambda_nH^{1/2}(\mathbb{R}^3, \mathbb{C}^4).
\end{split}\end{equation}
Since all the operators corresponding to the quadratic forms involved in \eqref{with root 2} commute with $\mathcal F_0^{-1}\mathfrak P\mathcal F_0$, it follows from \eqref{direct integral decomposition} that for almost all $\mathfrak{P}$
\begin{equation}\label{reduced assumption 2}
\langle\widetilde{\mathcal{H}}_{Z, 2}^{D_2, E_2, \mathfrak{P}}\psi, \psi\rangle_*\geqslant C_1\langle\sum_{n\in Z_2}\sqrt{\widehat{\mathbf{p}}_n^2 + m_n^2}\mathcal{F}\psi, \mathcal{F}\psi\rangle_*- C_2\|\psi\|_*^2
\end{equation}
holds for every $\psi\in \mathfrak{D}^{\mathfrak{P}}$, where $\widehat{\mathbf{p}}_n$ are defined in \eqref{p_n Fourier 2}. Thus $\mu$ is semibounded from below. Since by \eqref{p_n Fourier 2}
\begin{equation*}\label{P as sum 2}
\mathfrak{P}= \Big|\sum_{n\in Z_2}\widehat{\mathbf{p}}_n\Big|,
\end{equation*}
there exists $n\in Z_2$ such that 
\begin{equation*}\label{absolute value estimate 2}
|\widehat{\mathbf{p}}_n|\geqslant \frac{\mathfrak{P}}{N_2}
\end{equation*}
and hence
\begin{equation*}\label{to infinity 2}
\sum_{n\in Z_2}\sqrt{\widehat{\mathbf{p}}_n^2 + m_n^2}\geqslant \frac{\mathfrak{P}}{N_2}.
\end{equation*}
Thus the \rhs of \eqref{reduced assumption 2} tends to infinity as $\mathfrak{P}\to \infty$.
\end{proof}

\section{Absence of Gaps}

We are now ready to finish the proof of Theorem~\ref{HVZ theorem} by proving that
\begin{equation}\label{first inclusion}
\big[\varkappa(D, E), \infty\big)\subseteq \sigma(\mathcal{H}_N^{D, E}).
\end{equation}

Let us first fix a decomposition $Z$ on which the minimum is attained in \eqref{E(D)}.

Following the general strategy of \cite{JorgensWeidmann1973}, we will prove that for any irreducible representations $(D_1, E_1; D_2, E_2)\underset Z\prec (D, E)$ any
\begin{equation*}\label{lambda assumption}
\lambda\geqslant \varkappa_1(Z, D_1, E_1)+ \varkappa_2(Z, D_2, E_2)
\end{equation*}
belongs to $\sigma(\mathcal{H}_N^{D, E})$.
This will imply \eqref{first inclusion} according to the definition \eqref{bottom}.
Let
\begin{equation}\label{lambda_1}
\lambda_1:= \lambda- \varkappa_1(Z, D_1, E_1)\geqslant \varkappa_2(Z, D_2, E_2).
\end{equation}
We will use the notation and results of Section~\ref{easy part}.
The following lemma is a slight modification of Theorem~8.11 of \cite{JorgensWeidmann1973} and is proved along the same lines:

\begin{lemma}\label{generating subspace lemma}
Let $A$ be a selfadjoint operator in a Hilbert space $\mathfrak H$ and $U(\gamma)$ be a continuous representation of a compact group $\Gamma$ by unitary operators in $\mathfrak H$ such that $U(\gamma)\Dom A\subset \Dom A$ and $U(\gamma)A= AU(\gamma)$ for any $\gamma\in \Gamma$. Then for any irreducible (matrix) representation $D$ of $\Gamma$ the corresponding subspace $P^D\mathfrak H$ reduces $A$. For every $\lambda\in\sigma(A^D)$ where $A^D$ is the reduced operator and every $\varepsilon> 0$ there exists a $D$--generating subspace $G$ of $\Dom A$ such that
\begin{equation*}
\|Au- \lambda u\|\leqslant \varepsilon\|u\|,\ \textrm{for all}\ u\in G.
\end{equation*}
\end{lemma}

\begin{remark}
Recall that a subspace $G$ of $\mathfrak H$ is called {\em $D$--generating} if the operator $U(\gamma)\arrowvert G$ is unitary in $G$ for all $\gamma\in \Gamma$ and there exists an orthonormal base in $G$ such that for every $\gamma\in \Gamma$ the operator $U(\gamma)\arrowvert G$ is represented by the matrix $D(\gamma)$.
\end{remark}

\begin{proof}{ of Lemma~\ref{generating subspace lemma}}
Let $r$ be the dimension of the representation $D: \gamma\mapsto \big(D_{lk}(\gamma)\big)_{l, k= 1}^r$.
Let us introduce in $\mathfrak H$ the bounded operators $P_{lk}$ by
\begin{equation*}\label{P_lk}
P_{lk}:= r\int_\Gamma \overline{D_{lk}(\gamma)}U(\gamma)d\mu(\gamma), \quad l, k= 1, \dots, r,
\end{equation*}
where $\mu$ is the invariant probability measure on $\Gamma$.
It is shown in the proof of Theorem~8.11 of \cite{JorgensWeidmann1973} that $P_{ll}$ are orthogonal projections onto mutually orthogonal subspaces of $\mathfrak H$ and that
\begin{equation}\label{row decomposition}
P^D= \sum_{l= 1}^rP_{ll}.
\end{equation}
In fact, $P_{ll}$ is the projection on the subspace of function which belong to the $l^{th}$ row of the representation $D$.
Moreover, $P_{lk}$ is a partial isometry between $P_{kk}\mathfrak H$ and $P_{ll}\mathfrak H$. 
Since $\lambda\in\sigma(A^D)$, there exists a vector $u_0\in\Dom A^D$ such that
\begin{equation*}\label{u_0 choice}
\|A^Du_0- \lambda u_0\|\leqslant \varepsilon\|u_0\|.
\end{equation*}
It follows from \eqref{row decomposition} that there exists $l\in \{1, \dots, r\}$ such that $\|P_{ll}u_0\|\geqslant r^{-1}$. We can thus define $u_l:= P_{ll}u_0/\|P_{ll}u_0\|$ and then $u_k:= P_{kl}u_l$ for $k= 1, \dots, r$. The subspace $G$ spanned by $\{u_k\}_{k= 1}^r$ satisfies the statement of the lemma.
\end{proof}

Let
\begin{equation}\label{r_j}
r_j:= \dim(D_j\otimes E_j), \quad j= 1, 2.
\end{equation}
Since $\varkappa_1(Z, D_1, E_1)$ belongs to the spectrum of $\widetilde{\mathcal{H}}_{Z, 1}^{D_1, E_1}$ (see definition \eqref{E}), by Lemma~\ref{generating subspace lemma} we can choose a sequence of $(D_1\otimes E_1)$--generating subspaces $\{G_q\}_{q= 1}^\infty$ of $\Dom(\widetilde{\mathcal{H}}_{Z, 1}^{D_1, E_1})$ such that for all $q\in \mathbb N$
\begin{equation}\label{phi_q correct energy}
\big\|\widetilde{\mathcal{H}}_{Z, 1}^{D_1, E_1}\phi_q- \varkappa_1(Z, D_1, E_1)\phi_q\big\|_{\mathfrak{H}_{Z, 1}}\leqslant q^{-1}\|\phi_q\|_{\mathfrak{H}_{Z, 1}},\ \textrm{for all}\ \phi_q\in G_q.
\end{equation}
Analogously, for any $\mathfrak P\geqslant 0$ we can find a sequence $\{G^{\mathfrak P}_q\}_{q= 1}^\infty$ of $(D_2\otimes E_2)$--generating subspaces of $\Dom\widetilde{\mathcal{H}}_{Z, 2}^{D_2, E_2, \mathfrak{P}}$ such that
\begin{equation}\label{psi_q^P correct energy}
\big\|\widetilde{\mathcal{H}}_{Z, 2}^{D_2, E_2, \mathfrak{P}}\psi_q^{\mathfrak P}- \mu(\mathfrak P)\psi_q^{\mathfrak P}\big\|_*\leqslant q^{-1}\|\psi_q^{\mathfrak P}\big\|_*,\ \textrm{for all}\ \psi^{\mathfrak P}_q\in G_q^{\mathfrak P}.
\end{equation}
Moreover, we can choose a set of functions $\{\psi_{q, l}^{\mathfrak P}\}_{l= 1}^{r_2}$ in $\Dom\widetilde{\mathcal{H}}_{Z, 2}^{D_2, E_2, \mathfrak{P}}$ with 
\begin{equation*}\label{psi^P normalization}
\|\psi_{q, l}^{\mathfrak P}\|_*= 1
\end{equation*}
in such a way that for every $q\in\mathbb N$ and $l=1, \dots, r_2$ $\psi_{q, l}$ belongs to the $l^{th}$ row of the representation $(D_2\otimes E_2)$ and satisfies \eqref{psi_q^P correct energy}.
By Proposition~\ref{reduction for H^D,E}, Lemma~\ref{mu continuity lemma 2}, and \eqref{lambda_1} we can choose $\mathfrak P_0$ in such a way that
\begin{equation}\label{P_0 choice}
\mu(\mathfrak P_0)= \lambda_1.
\end{equation}

We choose $R_q> q$ so that \eqref{V_n decay} and \eqref{U_nj decay} hold true for all $n, j= 1, \dots, N$, $n< j$ with 
\begin{equation}\label{choice of epsilon}
\varepsilon:= q^{-1}(N_1+ 1)^{-1}N_2^{-1/2}C_1^{1/2}\big(C_2+ |\lambda_1|+ 2\big)^{-1/2},
\end{equation}
where $N_{1, 2}$ are the numbers of particles in $Z_{1, 2}$, and $C_{1, 2}$ are the constants in \eqref{assumption 1} for $j= 2$, and so that for some orthonormal base $\{\phi_{q, k}\}_{k= 1}^{r_1}$ of $G_q$
\begin{equation}\label{phi localization}
\bigg\|\Big(1- \prod_{j\in Z_1}I_{\{|\x_j|< R_q\}}\Big)\phi_{q, k}\bigg\|_{L_2(\mathbb R^{3N_1}, \mathbb C^{4^{N_1}})}\leqslant \frac{\nu_0}{4d_E^2r_1r_2},
\end{equation}
where $d_E$ is the dimension of $E$, $r_{1, 2}$ are defined in \eqref{r_j}, and the constant $\nu_0> 0$ depending only on $E, E_1, E_2$ will be specified later in the proof of Lemma~\ref{permutation projection lower bound lemma}.

By Assumption~\ref{square integrability assumption} and Lemma~\ref{mu continuity lemma 2}, we can choose a sequence of positive numbers $\{\delta_q\}_{q= 1}^\infty$ tending to zero in such a way that
\begin{equation}\label{delta_q condition 1}
\big|\mu(\mathfrak P)- \lambda_1\big|\leqslant q^{-1} \quad \textrm{for all}\quad \mathfrak P\in[\mathfrak P_0, \mathfrak P_0+ \delta_q],
\end{equation}
\begin{equation}\label{delta_q condition 2}
\frac1{2\pi^2}(\mathfrak P_0+ \delta_q)^2\delta_qC_{R_q}< q^{-2},
\end{equation}
where $C_{R_q}$ is the constant in \eqref{V and U integrability}, and
\begin{equation}\label{extra delta condition}
\frac1{2\pi^2}(\mathfrak P_0+ \delta_q)^2\delta_q\cdot\frac43\pi R_q^3< \frac{\nu_0^2}{16d_E^4r_1^2r_2^2}.
\end{equation}

Let us choose a function $f_q\in L_2(\mathbb R_+)$ with $\supp f_q\subset [\mathfrak P_0, \mathfrak P_0+ \delta_q]$ so that
\begin{equation}\label{condition on f_q}
\int_{\mathfrak P_0}^{\mathfrak P_0+ \delta_q}\big|f_q(\mathfrak P)\big|^2\mathfrak P^2d\mathfrak P= 1.
\end{equation}
Let
\begin{equation}\label{psi_q}\begin{split}
&\psi_{q, l}(\y_0, \dots, \y_{N_2- 1})\\ &:= \frac1{(2\pi)^{\frac32}}\int\limits_{\mathfrak P_0}^{\mathfrak P_0+ \delta_q}\int\limits_{S^2}e^{i\mathfrak P\omega\y_0}f_q(\mathfrak P)\psi_{q, l}^{\mathfrak P}(\omega, \y_1, \dots, \y_{N_2- 1})\mathfrak P^2d\omega d\mathfrak P,
\end{split}\end{equation}
where $\{\y_0, \dots, \y_{N_2- 1}\}$ and $\{\x_n\}_{n\in Z_2}$ are related by \eqref{new coordinates Z_2} and \eqref{old coordinates}. It follows from \eqref{condition on f_q} and the choice of $\psi_{q, l}^{\mathfrak P}$ that
\begin{equation}\label{psi_ql normalization}
\|\psi_{q, l}\|_{\mathfrak H_{Z, 2}}= 1, \quad l= 1, \dots, N_2,
\end{equation}
and that $\psi_{q, l}$ belongs to the $l^{th}$ row of $(D_2\otimes E_2)$.
Clearly the linear subspace $\widetilde G_q$ spanned by $\{\psi_{q, l}\}_{l= 1}^{r_2}$ is a $(D_2\otimes E_2)$--generating subspace of $\Dom \widetilde{\mathcal H}_{Z, 2}^{D_2, E_2}$.
\begin{lemma}\label{density boundedness lemma}
For any $n\in Z_2$ and $\psi\in \widetilde G_q$ with $\|\psi\|= 1$ the one--particle density
\begin{equation*}\label{1-particle density}
\rho_{\psi, n}(\x_n):= \int_{\mathbb R^{3N_2- 3}}\big|\psi(\x_{n_1}, \dots, \x_{n_{N_2}})\big|^2(d\x_{n_1}\cdots d\x_{n_{N_2}})/d\x_n
\end{equation*}
satisfies
\begin{equation*}\label{density boundedness}
\|\rho_{\psi, n}\|_{L_\infty(\mathbb R^3)}\leqslant \frac1{2\pi^2}(\mathfrak P_0+ \delta_q)^2\delta_q.
\end{equation*}
\end{lemma}
\begin{proof}{}
By \eqref{psi_q}
\begin{equation}\label{calculation}\begin{split}
&\|\rho_{\psi, n}\|_{L_\infty(\mathbb R^3)}\leqslant (2\pi)^{-3/2}\|\widehat\rho_{\psi, n}\|_{L_1(\mathbb R^3)}\\ &= \frac1{(2\pi)^6}\int_{\mathbb R^3}\bigg|\int_{\mathbb R^{3N_2}}\int_{\mathfrak P_0}^{\mathfrak P_0+ \delta_q}\int_{S^2}\int_{\mathfrak P_0}^{\mathfrak P_0+ \delta_q}\int_{S^2}e^{-i\mathbf p(\y_0+ \mathbf r_n)}e^{-i\mathfrak P\omega\y_0}\overline{f_q(\mathfrak P)}\\ &\times\psi_q^{\mathfrak P*}(\omega, \y_1, \dots, \y_{N_2- 1})e^{i\widetilde{\mathfrak P}\widetilde\omega\y_0}f_q(\widetilde{\mathfrak P})\psi_q^{\widetilde{\mathfrak P}}(\widetilde\omega, \y_1, \dots, \y_{N_2- 1})\mathfrak P^2\widetilde{\mathfrak P}^2\\ &\times d\widetilde\omega\, d\widetilde{\mathfrak P}\, d\omega\, d\mathfrak P\, d\y_0\, d\y_1\cdots d\y_{N_2- 1}\bigg|d\mathbf p,
\end{split}\end{equation}
where $\mathbf r_n:= \x_n- \y_0$, see \eqref{old coordinates}. Integrating the \rhs of \eqref{calculation} in $\y_0$ we obtain $(2\pi)^3\delta(\mathbf p+ \mathfrak P\omega- \widetilde{\mathfrak P}\widetilde\omega)$ from all the factors involving $\y_0$. Estimating the absolute value of the integral by the integral of absolute value and taking into account that $\int\delta(\mathbf p+ \dots)d\mathbf p= 1$ we get
\begin{equation}\label{calculation 2}\begin{split}
&\|\rho_{\psi, n}\|_{L_\infty(\mathbb R^3)}\leqslant \frac1{(2\pi)^3}\int_{\mathbb R^{3N_2- 3}}\int_{\mathfrak P_0}^{\mathfrak P_0+ \delta_q}\int_{S^2}\int_{\mathfrak P_0}^{\mathfrak P_0+ \delta_q}\int_{S^2}\big|f_q(\mathfrak P)\big|\big|f_q(\widetilde{\mathfrak P})\big|\\ &\times\big|\psi_q^{\mathfrak P}(\omega, \y_1, \dots, \y_{N_2- 1})\big|\big|\psi_q^{\widetilde{\mathfrak P}}(\widetilde\omega, \y_1, \dots, \y_{N_2- 1})\big|\mathfrak P^2\widetilde{\mathfrak P}^2\\ &\times d\widetilde\omega\, d\widetilde{\mathfrak P}\, d\omega\, d\mathfrak P\, d\y_1\cdots d\y_{N_2- 1}\leqslant \frac1{(2\pi)^3}4\pi(\mathfrak P_0+ \delta_q)^2\delta_q,
\end{split}\end{equation}
where at the last step we have used Schwarz inequality and $\|\psi\|= 1$. The formal calculation \eqref{calculation} --- \eqref{calculation 2} is justified by the fact that the integral over $\mathbb R^{3N_2}$ can be considered as a limit of integrals over expanding finite volumes, since $\psi\in L_2(\mathbb R^{3N_2})$.
\end{proof}
\begin{corollary}\label{L_2 times rho corollary}
For any $W\in L_2(\mathbb R^3)$, $n\in Z_2$, and $\psi\in \widetilde G_q$ with $\|\psi\|= 1$ we have
\begin{equation*}\label{W psi estimate}
\int_{\mathbb R^{3N_2}}\big|W(\x_n)\psi(\x_{n_1}, \dots, \x_{n_{N_2}})\big|^2d\x_{n_1}\cdots d\x_{n_{N_2}}\leqslant \frac1{2\pi^2}(\mathfrak P_0+ \delta_q)^2\delta_q\|W\|^2.
\end{equation*}
\end{corollary}

Let $F_q$ be a subspace of $\mathfrak H_N$ spanned by the functions
\begin{equation}\label{varphi_q}\begin{split}
\varphi_{q, k, l}(\x_1, \dots, \x_N):= \phi_{q, k}(\x_j: j\in Z_1)\otimes\psi_{q, l}(\x_n: n\in Z_2),\\ k= 1, \dots, r_1, \quad l= 1, \dots, r_2,
\end{split}\end{equation}
where $\{\phi_{q, k}\}_{k= 1}^{r_1}$ and $\{\psi_{q, l}\}_{l= 1}^{r_2}$ are orthonormal bases of $G_q$ and $\widetilde G_q$, respectively. We obviously have $\|\varphi_{q, k, l}\|_{L_2(\mathbb R^{3N}, \mathbb C^{4^N})}= 1$.
\begin{lemma}\label{unsymmetrized Weyl lemma}
For any $q\in\mathbb N$ $F_q\subset \Dom \mathcal H_N$. For any $\varphi\in F_q$
\begin{equation*}\label{on F_q}
\big\|(\mathcal H_N- \lambda)\varphi\big\|\leqslant 5q^{-1}r_1^{1/2}r_2^{1/2}\|\varphi\|.
\end{equation*}
\end{lemma}

\begin{proof}{}
It is enough to show that the functions \eqref{varphi_q} belong to $\Dom \mathcal H_N$ and satisfy
\begin{equation}\label{unsymmetrized Weyl}
\big\|(\mathcal H_N- \lambda)\varphi_{q, k, l}\big\|_{L_2(\mathbb R^{3N}, \mathbb C^{4^N})}\leqslant 5q^{-1}.
\end{equation}
Indeed, by triangle and Cauchy inequalities for
\begin{equation}\label{phi in base}
\varphi= \sum_{k= 1}^{r_1}\sum_{l= 1}^{r_2}c_{kl}\varphi_{q, k, l}
\end{equation}
we have
\begin{equation*}\begin{split}
\big\|(\mathcal H_N- \lambda)\varphi\big\|&\leqslant \sum_{k= 1}^{r_1}\sum_{l= 1}^{r_2}|c_{kl}|\big\|(\mathcal H_N- \lambda)\varphi_{q, k, l}\big\|\\ &\leqslant \underset{k, l}\sup \big\|(\mathcal H_N- \lambda)\varphi_{q, k, l}\big\|r_1^{1/2}r_2^{1/2}\|\varphi\|.
\end{split}\end{equation*}

The domain of $\mathcal H_N$ can be characterized as the set of functions $\xi$ from the form domain $\underset{n= 1}{\overset N\otimes}\Lambda_nH^{1/2}(\mathbb{R}^3, \mathbb{C}^4)$ on which the sesquilinear form $\langle\mathcal H_N\xi, \cdot\rangle$ is a bounded linear functional in $\mathfrak{H}_N$.
Functions \eqref{varphi_q} belong to $\underset{n= 1}{\overset N\otimes}\Lambda_nH^{1/2}(\mathbb{R}^3, \mathbb{C}^4)$ by construction.
By \eqref{H_N'}, \eqref{H_Z_1}, and \eqref{H_Z_2} we have
\begin{equation}\label{clusters+ intercluster interaction}
\mathcal H_N= \mathcal H_{Z, 1}+ \mathcal H_{Z, 2}+ \Lambda^N\bigg(\sum_{n\in Z_2}V_n+ \sum_{\substack{n< j\\ n\# j}}U_{nj}\bigg)\Lambda^N.
\end{equation}
The sesquilinear forms $\langle(\mathcal H_{Z, 1}+ \mathcal H_{Z, 2})\varphi_{q, k, l}, \cdot\rangle$ are bounded linear functionals over $L_2(\mathbb{R}^{3N}, \mathbb{C}^{4^N})$, since $\phi_{q, k}\in \Dom(\widetilde{\mathcal{H}}_{Z, 1}^{D_1, E_1})$ and $\psi_{q, l}\in \Dom \widetilde H_{Z, 2}^{D_2, E_2}$.
Moreover, by \eqref{phi_q correct energy}
\begin{equation*}\label{first correct energy}
\Big\|\big(\mathcal H_{Z, 1}- \varkappa_1(Z, D_1, E_1)\big)\varphi_{q, k, l}\Big\|= \Big\|\big(\widetilde{\mathcal{H}}_{Z, 1}^{D_1, E_1}- \varkappa_1(Z, D_1, E_1)\big)\phi_{q, k}\Big\|\leqslant q^{-1},
\end{equation*}
and by \eqref{psi_q^P correct energy}, \eqref{P_0 choice}, \eqref{delta_q condition 1}, \eqref{psi_q}, and \eqref{psi_ql normalization}
\begin{equation}\label{second correct energy}
\big\|(\mathcal H_{Z, 2}- \lambda_1)\varphi_{q, k, l}\big\|= \big\|(\widetilde{\mathcal{H}}_{Z, 2}^{D_2, E_2}- \lambda_1)\psi_{q, l}\big\|\leqslant 2q^{-1}.
\end{equation}
In view of \eqref{clusters+ intercluster interaction}---\eqref{second correct energy} and \eqref{lambda_1}, to prove that $\varphi_{q, k, l}\in \Dom \mathcal H_N$ and that \eqref{unsymmetrized Weyl} holds true it is enough to obtain that
\begin{equation}\label{intercluster interaction}
\bigg\|\bigg(\sum_{n\in Z_2}V_n+ \sum_{\substack{n< j\\ n\# j}}U_{nj}\bigg)\varphi_{q, k, l}\bigg\|\leqslant 2q^{-1}.
\end{equation}
To do this, we first note that by \eqref{V_n decay}, \eqref{U_nj decay}, and Cauchy inequality
\begin{equation}\label{V_n outside}\begin{split}
&\bigg\|\bigg(\sum_{n\in Z_2}V_nI_{\{|\x_n|> R_q\}}+ \sum_{\substack{n< j\\ n\# j}}U_{nj}I_{\{|\x_n- \x_j|> R_q\}}\bigg)\varphi_{q, k, l}\bigg\|\\ &\leqslant \varepsilon(N_1+ 1)\sum_{n\in Z_2}\big\||D_n|^{\frac12}\psi_{q, l}\big\|\leqslant \varepsilon(N_1+ 1)N_2^{\frac12}\bigg(\sum_{n\in Z_2}\big\||D_n|^{\frac12}\psi_{q, l}\big\|^2\bigg)^{\frac12}.
\end{split}\end{equation}
By \eqref{assumption 1}, \eqref{psi_ql normalization}, and \eqref{second correct energy},
\begin{equation}\label{technical stuff}\begin{split}
\sum_{n\in Z_2}\|D_n^{1/2}\psi_{q, l}\|^2&\leqslant C_1^{-1}\Big(\big\|(\widetilde{\mathcal H}_{Z, 2}^{D_2, E_2}- \lambda_1)\psi_{q, l}\big\|+ C_2+ |\lambda_1|\Big)\\ &\leqslant C_1^{-1}\big(C_2+ |\lambda_1|+ 2q^{-1}\big).
\end{split}\end{equation}
Thus by \eqref{V_n outside}, \eqref{technical stuff} and \eqref{choice of epsilon} for $q\geqslant 1$ we obtain
\begin{equation}\label{outside estimate}
\bigg\|\bigg(\sum_{n\in Z_2}V_nI_{\{|\x_n|> R_q\}}+ \sum_{\substack{n< j\\ n\# j}}U_{nj}I_{\{|\x_n- \x_j|> R_q\}}\bigg)\varphi_{q, k, l}\bigg\|\leqslant q^{-1}.
\end{equation}
Now the scalar functions
\begin{equation}\label{cut off in ball}
V_{n, q}(\x):= \big|V_n(\x)\big|I_{\{|\x|\leqslant R_q\}}(\x)\quad \textrm{and}\quad U_{nj, q}(\x):= \big|U_{nj}(\x)\big|I_{\{|\x|\leqslant R_q\}}(\x)
\end{equation}
are square integrable by \eqref{V and U integrability}. By Corollary~\ref{L_2 times rho corollary}, for $n\in Z_2$
\begin{equation}\label{long calculation for V_n}
\|V_{n, q}\varphi_{q, k, l}\|^2= \|V_{n, q}\psi_{q, l}\|^2\leqslant \frac1{2\pi^2}\delta_q(\mathfrak P_0+ \delta_q)^2\|V_{n, q}\|_{L_2(\mathbb R^3)}^2
\end{equation}
and for $n< j$, $n\# j$
\begin{equation}\label{try to be short}
\|U_{nj, q}\varphi_{q, k, l}\|^2\leqslant \underset{\mathbf{z}\in \mathbb R^3}\sup \big\|U_{nj, q}(\cdot- \mathbf z)\psi_{q, l}\big\|^2\leqslant \frac1{2\pi^2}\delta_q(\mathfrak P_0+ \delta_q)^2\|U_{nj, q}\|_{L_2(\mathbb R^3)}^2.
\end{equation}
Hence by \eqref{cut off in ball}, \eqref{long calculation for V_n}, \eqref{try to be short}, \eqref{V and U integrability}, and \eqref{delta_q condition 2}
\begin{equation}\label{inside estimate}\begin{split}
\bigg\|\bigg(\sum_{n\in Z_2}V_nI_{\{|\x_n|\leqslant R_q\}}+ \sum_{\substack{n< j\\ n\# j}}U_{nj}I_{\{|\x_n- \x_j|\leqslant R_q\}}\bigg)\varphi_{q, k, l}\bigg\|\leqslant q^{-1}.
\end{split}\end{equation}
It remains to add \eqref{outside estimate} and \eqref{inside estimate} to obtain \eqref{intercluster interaction}, finishing the proof of the lemma.
\end{proof}

The subspace $F_q$ spanned by the functions \eqref{varphi_q} is $D_1\otimes E_1\otimes D_2\otimes E_2$--generating. Since it is a sum of $D_1\otimes D_2$--generating subspaces, and $(D_1, E_1; D_2, E_2)\underset Z\prec (D, E)$, $F_q$ contains some nontrivial $D$--generating subspace. Hence the subspace $K_q:= P^DF_q$ is not equal to $\{0\}$ and is contained in $F_q$. 
\begin{lemma}\label{permutation projection lower bound lemma}
There exists a constant $C_E> 0$ such that for every $q$
\begin{equation}\label{permutation projection lower bound}
\|P^E\varphi\|\geqslant C_E\|\varphi\|, \quad \textrm{for all}\quad \varphi\in F_q.
\end{equation}
\end{lemma}

\begin{proof}{}
The projector \eqref{permutation projection} can be written as
\begin{equation}\label{permutation projection decomposition}
P^E= \frac{d_E}{h_\Pi}\sum_{\pi\in \Pi^Z_1\times \Pi^Z_2}\overline{\xi_E(\pi)}\pi+ \frac{d_E}{h_\Pi}\sum_{\pi\in \Pi\setminus(\Pi^Z_1\times \Pi^Z_2)}\overline{\xi_E(\pi)}\pi.
\end{equation}
We will denote the first term in \eqref{permutation projection decomposition} by $Q^E$, and the second by $R^E$. Then
\begin{equation}\label{quadratic form with Q and R}
\|P^E\varphi\|^2= \langle \varphi, P^E\varphi\rangle= \langle \varphi, Q^E\varphi\rangle+ \langle \varphi, R^E\varphi\rangle.
\end{equation}
Relation $(D_1, E_1; D_2, E_2)\underset Z\prec (D, E)$ implies that the representation $E\arrowvert \Pi^Z_1\times \Pi^Z_2$ is unitarily equivalent to a sum $\underset{i= 0}{\overset k\oplus}n_iE^{(i)}$, where $n_i> 0$ are multiplicities of the irreducible representations $E^{(i)}$ of the group $\Pi^Z_1\times \Pi^Z_2$ with $E^{(0)}= E_1\otimes E_2$. For the corresponding characters this gives
\begin{equation*}\label{character relation}
\xi_E(\pi)= \sum_{i= 0}^kn_i\xi^{(i)}(\pi), \quad\textrm{for all}\quad \pi\in \Pi^Z_1\times \Pi^Z_2.
\end{equation*}
Hence
\begin{equation*}\label{Q_E}
Q^E= \sum_{i= 0}^k\nu_iP_i,
\end{equation*}
where $\nu_i> 0$ and $P_i$ is the projector corresponding to the representation $E^{(i)}$. By construction, $P_0\varphi= \varphi$ for any $\varphi\in F_q$, hence $P_i\varphi= 0$ for $i= 1, \dots, k$. Thus for any $\varphi\in F_q$
\begin{equation}\label{P^E varphi}
\langle\varphi, Q^E\varphi\rangle= \nu_0\|\varphi\|^2, \quad \nu_0> 0.
\end{equation}
We will now estimate the second term on the \rhs of \eqref{quadratic form with Q and R}.
For any $n\in Z_2$ and any $\psi\in \widetilde G_q$ with $\|\psi\|= 1$ by Corollary~\ref{L_2 times rho corollary} and \eqref{extra delta condition} we have
\begin{equation}\label{psi in ball}
\|I_{\{|\x_j|< R_q\}}\psi\|^2\leqslant \frac{\nu_0^2}{16d_E^4r_1^2r_2^2}.
\end{equation}
For any functions \eqref{varphi_q} and any $\pi\in \Pi$ inequality \eqref{phi localization} implies that
\begin{equation*}\label{cross-terms on basis}
\big|\langle\varphi_{q, k, l}, \pi\varphi_{q, \widetilde k, \widetilde l}\rangle\big|\leqslant \langle\prod_{j\in Z_1}I_{\{|\x_j|< R_q\}}|\varphi_{q, k, l}|, \pi|\varphi_{q, \widetilde k, \widetilde l}|\rangle_{L_2(\mathbb R^{3N})}+ \frac{\nu_0}{4d_E^2r_1r_2}.
\end{equation*}
Now if $\pi\in \Pi\setminus(\Pi^Z_1\times\Pi^Z_2)$, then there exists $j_0\in Z_1$ such that $\pi j_0\in Z_2$. Then by \eqref{psi in ball}
\begin{equation*}\label{cross-terms 2}
\langle\prod_{j\in Z_1}I_{\{|\x_j|< R_q\}}|\varphi_{q, k, l}|, \pi|\varphi_{q, \widetilde k, \widetilde l}|\rangle\leqslant \langle|\varphi_{q, k, l}|, I_{\{|\x_{j_0}|< R_q\}}\pi|\varphi_{q, \widetilde k, \widetilde l}|\rangle\leqslant \frac{\nu_0}{4d_E^2r_1r_2}.
\end{equation*}
Thus
\begin{equation}\label{all cross terms}
\big|\langle\varphi_{q, k, l}, \pi\varphi_{q, \widetilde k, \widetilde l}\rangle\big|\leqslant \frac{\nu_0}{2d_E^2r_1r_2}, \quad \pi\in \Pi\setminus(\Pi^Z_1\times\Pi^Z_2).
\end{equation}
Any $\varphi\in F_q$ can be written as \eqref{phi in base}.
By \eqref{all cross terms} and Cauchy inequality for any $\pi\in \Pi\setminus(\Pi^Z_1\times\Pi^Z_2)$
\begin{equation}\label{R term control}
\big|\langle\varphi, \pi\varphi\rangle\big|\leqslant \sum_{k, l, \widetilde k, \widetilde l}|c_{kl}||c_{\widetilde k\widetilde l}|\big|\langle\varphi_{q, k, l}, \pi\varphi_{q, \widetilde k, \widetilde l}\rangle\big|\leqslant \frac{\nu_0}{2d_E^2}\|\varphi\|^2.
\end{equation}
Since the number of elements of $\Pi\setminus(\Pi^Z_1\times\Pi^Z_2)$ does not exceed $d_\Pi$ and for any $\pi$ $\big|\xi_E(\pi)\big|< d_E$ as a trace of unitary matrix of dimension $d_E$, \eqref{R term control} implies that
\begin{equation*}\label{R^E estimate}
\big|\langle\varphi, R^E\varphi\rangle\big|\leqslant \nu_0\|\varphi\|^2/2.
\end{equation*}
By \eqref{quadratic form with Q and R} and \eqref{P^E varphi} we conclude that \eqref{permutation projection lower bound} holds with $C_E= \sqrt{\nu_0/2}$.
\end{proof}

Lemmata~\ref{unsymmetrized Weyl lemma} and \ref{permutation projection lower bound lemma} imply that $L_q:= P^EK_q$ is a nontrivial subspace of $\Dom \mathcal{H}_N^{D, E}$ and for every $f= P^E\varphi\in L_q$
\begin{equation*}\label{Weyl relation}
\big\|(\mathcal H_N^{D, E}- \lambda)f\big\|\leqslant \big\|(\mathcal H_N- \lambda)\varphi\big\|\leqslant 5q^{-1}r_1^{\frac12}r_2^{\frac12}\|\varphi\|\leqslant 5q^{-1}r_1^{\frac12}r_2^{\frac12}C_E^{-1}\|f\|, \quad q\in\mathbb N.
\end{equation*}
This implies that $\lambda\in \sigma(\mathcal H_N^{D, E})$, and thus finishes the proof of Theorem~\ref{HVZ theorem}.
\paragraph{Acknowledgement.} The author was supported by the DFG grant SI 348/12--2. Part of this work was done during the stay at the Erwin Schr\"odinger Institute, Vienna.

\bibliographystyle{plain}
\bibliography{Bericht}
\Addresses
\end{document}